\documentclass[12pt,fleqn]{iopart}
\usepackage[english]{babel}
\usepackage[pdfpagelabels,plainpages=false]{hyperref}   
\usepackage{amssymb}
\usepackage{braket}
\usepackage{bbm}
\usepackage{fancyhdr}
\usepackage{units}
\usepackage{dsfont}

\usepackage{amsopn}
\usepackage{iopams}

\usepackage{graphicx}
\usepackage{color}
\usepackage[T1]{fontenc}
\usepackage[utf8]{inputenc}
\usepackage{amsthm}
\usepackage{bbm}
\usepackage{float}

\newtheorem{theorem}{Theorem}
\newtheorem*{theorem*}{Theorem}

\newtheorem{Lemma}[theorem]{Lemma}

\newtheorem{corollary}[theorem]{Corollary}

\newtheorem{example}[theorem]{Example}

\usepackage{bbold}

\bibliography{vancouver}

\begin{document}
		\title[Signatures of Lattice Excitations in Quantum Channels]{Signatures of Lattice Excitations in Quantum Channels: Limit of Parent Hamiltonians}
		\author{Beno\^it Descamps $^{1,2}$}
		
		\address{$^{1}$Vienna Center for Quantum Technology, University of Vienna, Boltzmanngasse 5, 1090 Wien, Austria}
		\address{$^{2}$Ghent University, Krijgslaan 281, 9000 Gent, Belgium}

	\date{\today}
	\begin{abstract}
We prove that every injective Matrix Product State is the unique ground state of a simple hopping theory. We start by studying the low energy spectrum of parent Hamiltonians of injective Matrix Product States in a particular long range and system size limit under the validity of an asymptotic regime with low particle density. We show that in this limit a natural first quantization arises. This allows us to compute a specific type of low energy spectrum. This spectrum depends solely on the properties of a quantum channel, i.e. transfer matrix of the ground state, and not on any other details of the ground-state.
We also review normal quantum channels for which the expression is more simplified. 

The construction possibly has some interesting uses for the study of quantum and classical Markov processes which we briefly expose.
As an application, we revisit the notion of (many-body)-localization with our framework. Our calculations revealed that translational invariant Matrix Product States can be interpreted as a stationary sea of particles. As a next step rather than starting from some local Hamiltonian with random potentials, we consider fluctuations of the local tensors of a continuous one-parameter family of Matrix Product States. Localization in 1-dimension, is then understood from a simple study of spectral and mixing properties of finite dimensional quantum channels. \\

	\end{abstract}
	
	\pacs{}
	\maketitle

\section{Introduction}
The general study of strongly interacting quantum many body systems requires a large amount of parameters. It appears, however, that a successful description of such systems is still possible in a much smaller manifold. Matrix product
states (MPS), originating from the density matrix re-normalization group
(DMRG) algorithm \cite{chapter 4: Verstraete Cirac,chapter 4: StevenWhite} , have turned out to be one of these particular manifolds used to describe the ground state of these systems.
Further study of the manifold laid down some ground work for studying lower energy properties \cite{chapter 4: Haegeman QGross,chapter 4: Haegeman PostMPS}.
Some argumentation can be made for the mathematical and physical validity of the methods \cite{chapter 4: Haegeman ElementExc}. However, we argue at the beginning of this work that this method is heavily biased by the spectral properties of the transfer matrix of the ground state. 
While every numerical method has its own bias, this opens the question as to which condensed matter properties hide inside Quantum Channels, i.e. transfer matrix, or trace preserving completely positive operators.

Recalling  K\"all\'en -Lehmann spectral representation in quantum field theory \cite{chapter 4: Kaller, chapter 4: Lehmann}, which connects 2-point correlations with free propagators, it is not completely unexpected that indeed the transfer matrix contains some information.

Apart from exactly solvable and integrable models or Lieb-Robinson bound arguments based on the structure of the dispersion relation \cite{chapter 4: Haegeman ElementExc}, few models can be found with a clear finite-particle pictures for low excitations. 
In this work, we propose an interesting approach to this problem. 
We start from a system of size $N$ with a dynamic described by a parent Hamiltonian, of an injective matrix product state, with finite range $L$. We then study the limit of $N,L\to \infty$ and possible re normalisation of the transfer matrix $\Gamma(L)$ under the additional restriction that an asymptotic regime for low particle density remains valid.

Under this limit, a clear particle-picture appears, which is completely related to the transfer matrix. Even-more so these excitations do not depend on any other microscopic properties of the ground-state.

In 1992, following the seminal work by Affleck, Kennedy, Lieb and Tasaki \cite{chapter 4: AKLT}, Fannes, et.al. \cite{Fannes} showed that Matrix Product States are the ground states of gapped local Hamiltonians. In this work, inspiring from this construction and considering an additional limit, we show that not only they are the ground-states of gapped Hamiltonians, but also this Hamiltonian describes the hopping bosons in one-dimension.

This work is divided into two distinct part. In the first part (\ref{section:bias},\ref{section:rewriting PH},\ref{section:FQ}), we elaborate on the construction. The details of the main theorem (\ref{Theorem: free Fermions}) can be found the appendix.
In the second part, we present some applications. First, we briefly describe some possible consequences and use for the field of classical and quantum Markovian dynamics (\ref{section:Stochastics}). Finally in the last part (\ref{section:Localization}), we revisit localization using the framework and philosophy of the work presented in the first part.
\newline

\section{From  K\"all\'en -Lehmann spectral representation to the tangent plane of Matrix Product States }
\label{section:bias}
Gunnar K\"all\'en \cite{chapter 4: Kaller} and  Harry Lehmann \cite{chapter 4: Lehmann} discovered independently that the two point-correlation in quantum field theory can be related to the free propagators,

$$\langle\Omega|T\phi(x)\phi(y)|\Omega\rangle = \int_0^\infty \frac{d\mu^2}{2\pi} \rho(\mu^2)\int \frac{d^4 p}{(2\pi)^4}\frac{1}{p^2-\mu^2+i\epsilon}e^{-ip(x-y)}$$
This is illustrated here for a scalar field theory, where $\rho(\mu^2)$ is a positive spectral density.
For a free scalar field, one can verify that the Fourier transform of the two-point correlation function, is inversely proportional to the dispersion relation,

$$\int d^4x e^{ipx}\langle\Omega|T\phi(x)\phi(y)|\Omega\rangle = \frac{i}{p^2-\mu^2+i\epsilon}$$
It is amazing to see how much information is stored simply inside the correlation function of the ground state of the system.

A first approach for studying low energy excitations in one-dimensional spin lattice makes use of the tangent plane of the ground-state MPS in the manifold \cite{chapter 4: Haegeman PostMPS}.
For a translational invariant MPS $|\phi\rangle$, 
\begin{eqnarray}
~~~~~~~~~~~~~~~~~~~~~~~~~~~~~~~~~~~~~~|\phi\rangle =\sum_{\vec{i}}\operatorname{Tr}\left(A^{i_1}\dots A^{i_N}\right)|\vec{i}\rangle
\end{eqnarray}
A simple proposition for an approximation of the basis of the m-particle subspace are the states,
\begin{eqnarray}
\label{eq:excAnsatz}
&|\psi\{B^{(i)}_\alpha,k_\alpha\}_{\alpha=1}^m\rangle =\sum_{n_1<\ldots<n_m} e^{i\vec{k}.\vec{n}}|\phi\{B^{(i)}_\alpha,n_\alpha\}\rangle \nonumber\\ 
&|\phi\{B^{(i)}_\alpha,j_\alpha\}_{\alpha=1}^m\rangle=\sum_{\vec{i}}\operatorname{Tr}\left(A^{i_1}\dots A^{i_{j-1}} B^{i_{n_1}}A^{i_{j+1}}\right. 
\left.\dots  B^{n_{j_m}}A^{i_{j+m}}\dots A^{i_N}\right)|\vec{i}\rangle
\end{eqnarray}
A sub-vector space can be considered, by choosing $B^{(i)}=A^{(i)}B$.

A particular approximation of the low energy spectrum is defined from the variational problem,
$$E_k=\min_{B}\sum_Z\frac{\langle \psi\{B,k\}| H_Z |\psi\{B,k\}\rangle}{\langle \psi\{B,k\}|\psi\{B,k\}\rangle},~~ \langle \phi|\psi\{B,k\}\rangle=0$$
In the case of frustration free Hamiltonians, such as parent Hamiltonians, $H_Z|\phi\rangle=0, H_Z\geq 0$, both numerator and denominator scale as $N$. The numerator is then finite while the denominator is given by the Fourier transform of the transfer matrix (\ref{eq:FourierTransfer}),
$$E_k\leq \frac{C}{\operatorname{Tr}\left(B^\dagger\mathcal{T}_k[\Gamma][B]\right)} ,~~\operatorname{Tr}\left(\rho B\right)$$
Generally maxima of the Fourier Transform of the 2-point correlation function, i.e. singular points in the  K\"all\'en - Lehmann spectral representation, represent the lowest excited particles.
Thus, we can argue that the tangent plane method is biased by the spectrum of the Fourier Transform of the Transfer Matrix $\Gamma$.

However, in the case that this particular ansatz is right, spectral properties transfer matrix, i.e. Fourier Transform, yields a certain amount of information about the spectrum of the Hamiltonian such as minima of the dispersion relation.

This particular observation leads us to the starting question of this work. Can we extract or make other predictions for the dispersion beyond this one-particle ansatz?
Is it possible to construct a family of models for which the spectrum can be completely extract from the transfer matrix?

We, indeed, manage to construct  a local Hamiltonian for which excitations depend purely on a quantum channel, and of which the m-particle eigenstates can be written using the m-particle basis (\ref{eq:excAnsatz}).
The procedure is based on a limit. It will also be interesting to study some consequence when the limit slightly breaks down.

\section{Rewriting Parent Hamiltonians}
\label{section:rewriting PH}
Parent Hamiltonians \cite{Fannes,chapter 2: Perez MPS rep} of Matrix Product States are defined as projectors onto the kernel of the vector space $\{\operatorname{Tr}\left(X A^{i_{j}}\dots A^{i_{j+L}}\right)\}$. Uniqueness of the ground state is assured if the map $X\to \operatorname{Tr}\left(X A^{i_{j}}\dots A^{i_{j+L}}\right)$ is injective.

In the case of our study, we notice that the following parametrization of the  parent Hamiltonian, pictured in figure (\ref{fig:parentH}), is of use to us.
Choose $C\in \mathcal{M}_{D^2}$, so that 
\begin{eqnarray}
\label{eq:parentH}
& H=\sum_j \left(\operatorname{Id}-h_{j,j+L}\right), \\
& h_{j,j+L}=\sum_{\vec{\alpha},\vec{\beta}}\operatorname{Tr}\left(C \left[A^{\alpha_j}\otimes \overline{A}^{\beta_j}\right] \dots \left[A^{\alpha_{j+L}}\otimes \overline{A}^{\beta_{j+L}}\right]\right)|\vec{\alpha}\rangle\langle \vec{\beta}| \nonumber
\end{eqnarray}
satisfies the parent Hamiltonian conditions. A simple choice is, 
\begin{eqnarray}
\label{eq:Correction}
~~~~~~~~~~~~~~~~~~~~~~~~~~~~~~~~~~~~~~~~C^{(i_1,i_2)}_{(j_1,j_2)}= \left(\langle j_2 i_2|\Gamma^L|i_1 j_1\rangle_{j_1j_2}^{i_1 i_2}\right)^{-1}
\end{eqnarray}

We must first point out a striking property of parent Hamiltonians which becomes apparent from equation (\ref{eq:parentH}) and figure (\ref{fig:parentH}).
Let us make use of the ansatz equation (\ref{eq:excAnsatz}) for the choice $B^{(i)}_\alpha=A^{(i)}_\alpha B$.
The action of the parent Hamiltonian, as given by  equations (\ref{eq:parentH},\ref{eq:Correction}) maps the $m$-particles wave equations onto a subspace spanned by $1,\ldots,m$-particles states.
Furthermore, and most importantly, the reader should see that the decomposition under the action, depends purely on the transfer matrix, $\Gamma$,
\begin{eqnarray}
\label{eq:Transfer}
~~~~~~~~~~~~~~~~~~~~~~~~~~~~~~~~~~~~~~~~~~~~~~\Gamma=\sum_j \overline{A}^{j}\otimes A^{j}
\end{eqnarray}

\begin{figure}[t!]

	\vspace{-0.1 cm} 
	\hbox{\hspace{+2 cm}\includegraphics[width=0.9\textwidth]{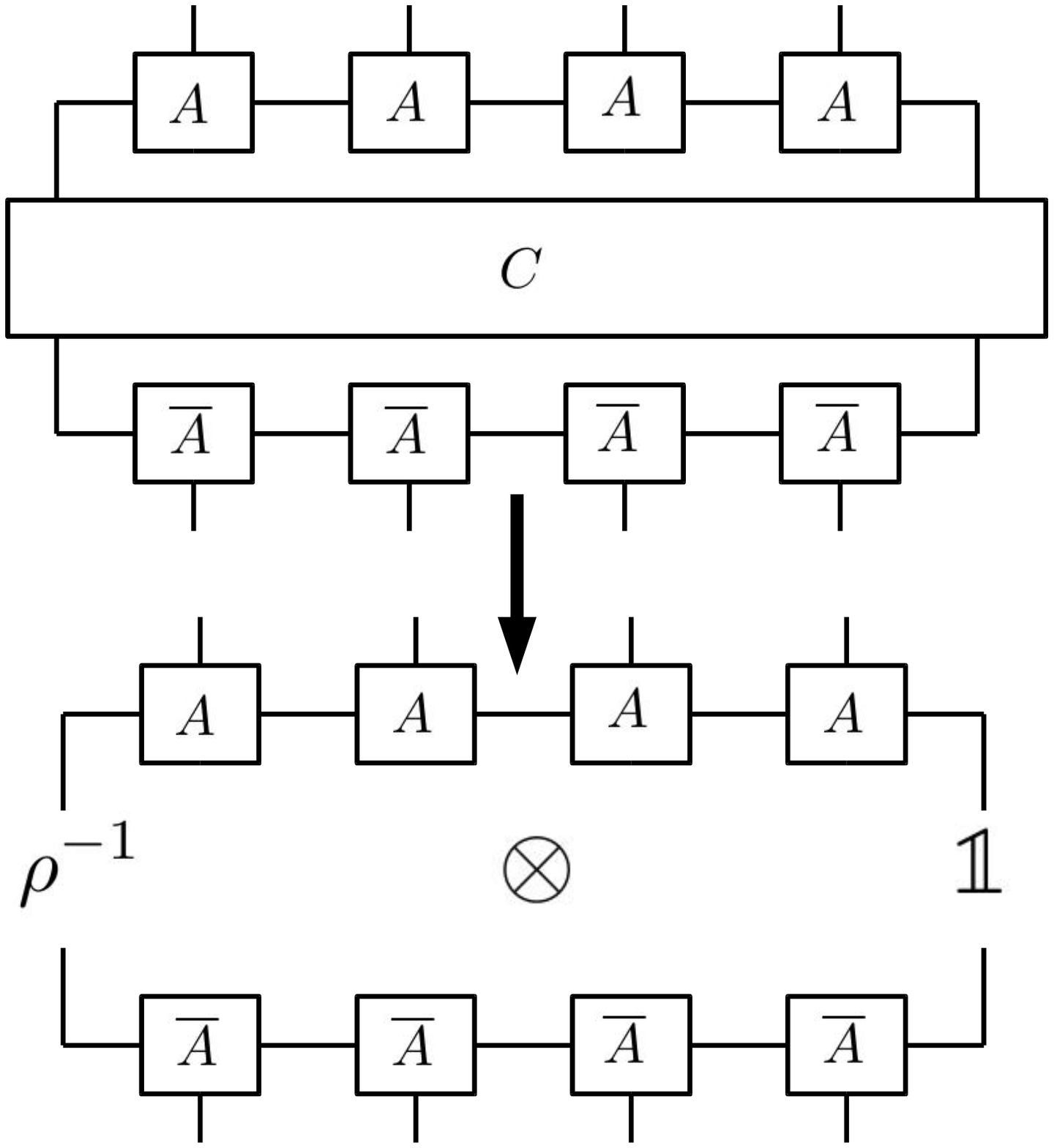}}
	\vspace{-1 cm}  
	\caption{Some useful representation for the parent Hamiltonians of a Matrix Product State. The choice of the tweaking matrix $C$ is not unique. A possible choice is given by equation (\ref{eq:Correction}). For this particular choice, the Matrix Product is an eigenvector with eigenvalue $1$ of this interaction term. Furthermore, they are projectors. However, it is unlikely that some low energy spectrum can be solved generically for any of such representation. One notices that all these representations convergence towards $\rho^{-1}\otimes \mathbb{1}$ in the limit of large range and system size.}
	\label{fig:parentH}
\end{figure}
\begin{figure}[t!]
	\vspace{-0.1 cm} 
	\centering \includegraphics[width=0.7\textwidth]{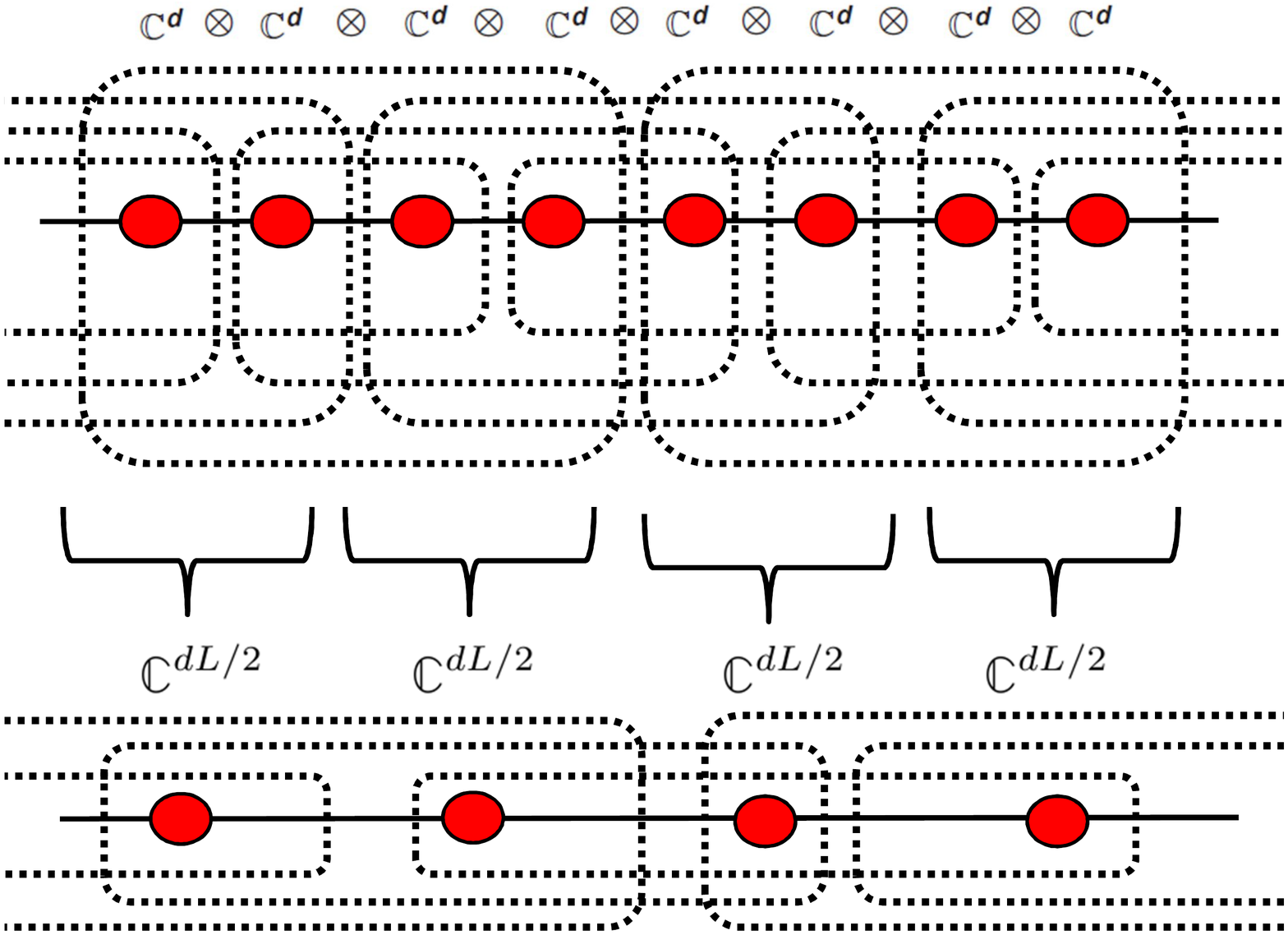}
	\vspace{+0.1 cm}  
	\caption{Since we take the limit $L\to \infty$, one might argue the physical relevance, as the Hamiltonian seems, at first sight, to lose its local property. By blocking sites, represented with red disks, and regrouping the interactions (dotted lines) we see however it is not the range of interactions which increases, but the dimension of the local Hilbert space (red disks in the second line). While the local Hilbert space grows in dimension, the interactions remain local, but unbounded, in the thermodynamic limit as $L/N \to \infty$.}
	\vspace{+0.1 cm}  
	\label{fig:limit}
\end{figure}
No other information about the ground state or particular Hamiltonian details, is encoded inside the dynamical spectrum besides the transfer matrix. 
We will see that we can relate the minima of the dispersion with the spectrum of the transfer matrix in the vicinity of the unit circle.\newline

\paragraph{Limit, $N,L\to \infty$}
\label{par:limit}
Clearly the true excitations generally do not satisfy the ansatz (\ref{eq:excAnsatz}). However, this is the case in the following limiting procedure.
First of all, by a choice of gauge $A^{(i)}\to S A^{(i)} S^{-1}$, we can make sure that the transfer matrix is trace preserving with a diagonal fixed point,
$$\Gamma^*[\mathbb{1}],~~\Gamma[\rho]=\rho=\operatorname{diag}(\vec{\lambda}),~~\operatorname{Tr}(\rho)=1$$
In the case of an injective MPS, it can be shown that the fixed point $\rho$ is unique and non-degenerate \cite{Fannes,chapter 2: Perez MPS rep}, even more so the spectrum always lies within the unit circle.
Second, we take the thermodynamic limit $L<N\to \infty$, while varying the range of each local interaction of the parent Hamiltonian. We want a meaningful notion of quasi-particles. As we add more particles, in the generic case, the energy should increase with the number of particles. Only in special cases should bound states appear. A reasonable choice is $L/N\to 0$.
It should be noted for this particular choice, there is still a notion of "local Hamiltonian" under rescaling of the metric with $L$, as shown in figure (\ref{fig:limit}).

Finally, we could renormalise the transfer matrix $\Gamma(L)$. However, we would like the interaction term to forget the correction $C$, in equation (\ref{eq:Correction}). 
Indeed our choice mostly converges as $C\to \rho^{-1}\otimes \mathbb{1}$. This is not the case if the second largest eigenvalue converges towards the unit circle as $O(1/L)$. So, we must keep in mind that there are some restriction on the re-normalization of the transfer.

As a side note as long as the spectrum of the renormalized super-operator stays within this $1-O(1/L)$-circle, the martingale method remains valid and the Hamiltonian is gapped \cite{Fannes}. \newline

\section{A Natural First Quantization: Hopping of Virtual Particles}
\label{section:FQ}
\subsection{One-particle excitations}
We saw earlier that the subspace spanned by $1,\ldots,m$-particles states remains stable under the dynamic. 
In the final result of theorem (\ref{Theorem: free Fermions}), we prove that they can be computed exactly within these subspaces.

In the introduction, we initially computed the energy using variational methods. As we desire to derive the whole spectrum, we seek an alternative approach. The goal is to demonstrate that for our limit, the states become exact eigenvectors of the Hamiltonian,
$$\frac{\|H|\kappa\rangle -\lambda |\kappa\rangle\|_2^2}{\langle \kappa |\kappa \rangle}\to 0 \rightarrow \mbox{ we set $|\kappa\rangle$ to be an eigenvector}$$
Let us first have a look at the limit without re-normalization of the transfer matrix.
We can propose the following representation for the one-particle states $|\psi\{X_a,k\}\rangle$ with momentum $k$,
\begin{eqnarray}
~~~~~~~~~~~~~~~~~~~~|\psi\{X_a,k\} =\sum_n e^{ikn} \sum_{\vec{i}}\operatorname{Tr}\left(A^{i_1}\dots A^{i_{n-1}} A^{i_{n}}X_a A^{i_{n+1}}\dots A^{i_N}\right)|\vec{i}\rangle
\end{eqnarray}

In Lemma (\ref{Lemma: one-particle}), we demonstrate that such states are indeed eigenstates, with energy that can calculated from an effective one-particle Hamiltonian $H_{\mbox{\scriptsize eff}}^{(1)}$. One of the central objects, contained in this effective Hamiltonian, is the Fourier Transform of the Transfer matrix $\mathcal{T}_k[\Gamma]$,
\begin{eqnarray}
\label{eq:FourierTransfer}
~~~~~~~~~~~~~~~~~~~~T_k[\Gamma]=& R_{\rho}+e^{ik}\frac{\Gamma}{2(\mathbb{1}-e^{ik}\Gamma)}\circ R_{\rho}+R_{\rho}\circ e^{-ik}\frac{\Gamma^*}{2(\mathbb{1}-e^{-ik}\Gamma^*)}
\end{eqnarray}
 This result is very evocative of the K\"all\'en -Lehmann spectral representation, \cite{chapter 4: Kaller, chapter 4: Lehmann}. Hence, all information about the spectrum can be derived from the spectral analysis of 
$\mathcal{T}_k[\Gamma]$.\newline

The following turns out to be extremely useful. One should notice the invariance of the one-particle states with momentum $k$ under the gauge transformation,
\begin{eqnarray}
\label{eq: main gauge}
~~~~~~~~~~~~~~~~~~~~~~~~~~~~~~~~~~A^{(i)}_a \rightarrow B^{(i)}_a = A^{(i)}X_a + e^{ik}Y_a A^{(i)} - A^{(i)}Y_a
\end{eqnarray}
However, this invariance breaks down when additional particles are taken into consideration.
The results for the many-particle wave-function are examined in the next section. We should already announce that in the long range and large system limit, we derive the eigenstates using the philosophy of the Bethe ansatz \cite{chapter 4: Bethe Ansatz}. Each eigenvector of the effective one-particle Hamiltonian, which is related to  $\mathcal{T}_k[\Gamma]$, is a particle with momentum $k$. These particles can theb be combined in a m-particle wave function. The energy is then the sum of the individual energy of each particle,
\begin{eqnarray}
H|\psi[\{X_1,\ldots,X_k;k_1,\ldots,k_m\}]\rangle \to \sum_j E[X_j,k_j]|\psi[\{X_1,\ldots,X_k;k_1,\ldots,k_m\}]\rangle 
\end{eqnarray}

\subsection{Breaking down of the Limit and Bound-States}

One could try to look "beyond" this simple one-particle picture. Generally, there exists modes consisting of a superposition of scattering particles. These states consisting of more than a single particle, sometimes have a lower energy than the other stable individual particles.

The framework, we have presented here, gives us the perfect opportunity for studying this phenomenon and even more so derive how this is reflected in the spectral properties of the transfer matrix.

For this, we need to renormalize the transfer matrix $\Gamma[L]$ with the range $L$. The idea would be to let the ground state become "critical". As the spectrum gets closer to the boundary of the unit circle, some particles, i.e. eigenvectors of $\mathcal{T}_k[\Gamma]$, become unstable and splits into other particles.

Technically, the ansatz breaks down with an error $\epsilon_1$, given by equation (\ref{eq:error1Part}). The new excited state consists of the 1-particle state in superposition with its fusion options $\left\{Y_{a_1},Y_{a_2}\right\}\to X_a$.

For simplicity, and as a sake of illustration, we look at normal unital trace-preserving completely positive operators,
$$\Gamma \Gamma^*=\Gamma^*\Gamma,~~\Gamma^*[\mathbb{1}]=\mathbb{1},~~\Gamma[X_a]=\lambda_{X_a} X_a$$

The eigenvectors are orthonormal and characterized by the structure tensor $f^a_{bc}$,
$$X_{a}X_{b}=\sum_c f^c_{ab} X_c $$

Normal unital quantum channels are particularly interesting here as the effective Hamiltonian is exactly the Fourier Transform of the Transfer matrix. The eigenvectors $X_a$ of the transfer matrix, are finite dimensional representations of the particles  with wave function $|\psi\{ X_a,k\}\rangle$,
\begin{eqnarray}
H |\psi\{ X_a,k\}\rangle \nonumber=\left(2L- 2\operatorname{Re}\frac{1}{1-e^{i\left[k-\phi_{X_a}\right]}|\lambda_{X_a}|}\right)|\psi\{ X_a,k\}\rangle 
\end{eqnarray}
with $\lambda_{X_a}=e^{-i\phi_{X_a}}|\lambda_{X_a}|$. Hence the phase of the eigenvalue of $X_a$ is interpreted as its rest-momentum.

Continuing our discussion at the beginning of this short section, we let the part of the spectrum converge towards the unit circle. The rate of convergence becomes clear by comparing $\langle\psi\{ \overline{X}_a,k\}|H |\psi\{ X_a,k\}\rangle $ with the generic one-particle energy, $\Delta_a= 2L - 2\operatorname{Re}\frac{1}{1-e^{i\left[k-\phi_{X_a}\right]}|\lambda_{X_a}|}-\langle\psi\{ \overline{X}_a,k\}|H |\psi\{ X_a,k\}\rangle$,

$$\Delta_a \approx  \sum_{\alpha,\beta} \frac{|f^\alpha_{a,\beta}|^2}{\sigma_{a,k}}\Big|\sum_{j=1}^L |\lambda_{X_\alpha}|^j |\lambda_{X_\beta}|^{L-j} e^{i[k-\phi_{X_\alpha}+\phi_{X_\beta}]j}\Big|^2$$
First, we see that $\Delta_a>0$, and so decreases the energy.

When two particles have approximately the same eigenvalues, $|\lambda_{X_\alpha}|\approx |\lambda_{X_\beta}|$, and approaches the unit circle as $|1-|\lambda_{X_\alpha}|\ |\approx \gamma (\ln L)/L$. The new contribution to the energy $\Delta_a$ peaks at $k=\phi_{X_\alpha}-\phi_{X_\beta}$ in the order of $L^{1-2\gamma}$.
Additionally, the contribution is proportional to the discrepancy of the one-particle sub-space. Hence for $\gamma<1/2$, the one-particle subspace is unstable under the dynamic.
The state $X_c$, resulting from the fusion of $X_\alpha$ and $X_{\beta}$, has a finite probability of decaying to these two particles. The bound state $ |\psi\{X_c,k\}\rangle +\sum_q c_q |\psi\{X_\alpha,k-q;X_\beta^\dagger,q\} $ emerges. This state has a minimum at $k=\phi_{X_\alpha}-\phi_{X_\beta}$.\newline

\paragraph{Example}
\label{par: example}

We illustrate this for the transfer matrix with eigenvectors, $\mathbb{1},\sigma^+, \sigma^-, \sigma^z$, with respective eigenvalues $1, \lambda_+, \lambda_-$ and $\lambda_z$.
We choose $\lambda_+=e^{i\phi}\lambda=\overline{\lambda}_-$,      $\lambda_z=\lambda^2$.
In the long range limit, the low energy spectrum is built by adding particles of the type $\sigma^+, \sigma^-$ or $\sigma^z$ to the wave function.

Choosing to renormalize the spectrum with the range $L$, yields a rescaling of the errors

$$\lambda(L)=e^{-\gamma \frac{\ln(L)}{L}}\rightarrow \epsilon_1(\sigma^+)=\epsilon_2(\sigma^+)\approx \frac{L^{1-2\gamma}}{\gamma \ln(L)}, \epsilon_1(\sigma^z)\approx 2\gamma \ln(L)L^{1-2\gamma}$$
For $\gamma=1/2$, the 1-particle subspace remains stable for $|\psi\{\sigma^+,k\}\rangle$, $|\psi\{\sigma^-,k\}\rangle$, and we get the bound state, 

$$|\psi^{(2)}\{Z,k\}\rangle =|\psi\{Z,k\}\rangle +\sum_q c_q^{(1)} |\psi\{\sigma^+,k-q;\sigma^-,q\}+c_q^{(2)} |\psi\{\sigma^-,k-q;\sigma^+,q\}$$

\subsection{Solving the Low Energy Spectrum in the Long Range Limit}

In 1931, Hans Bethe \cite{chapter 4: Bethe Ansatz} initiated a new technique which ultimately lead to a resolution of the Heisenberg spin-chain. His method, which was later given his name, the coordinate Bethe ansatz. In his paper, Bethe introduced $n$-particles wavefunctions of the form,

$$\psi(\vec{k}) = \sum\limits_{0\leq x_1 \leq \ldots \leq x_n \leq N} \sum\limits_{P\in\mathcal{S}_n}A(P)e^{i(k_{P_1} x_1 +\ldots+k_{P_N} x_N)}\phi(x_1,\ldots,x_n)$$
where $\phi(x_1,\ldots,x_n)$ is a wave-function in first quantization of spins at positions $x_1,\ldots,x_n$. 

Such plane-wave decomposition needs to satisfy a consistency condition,

$$A(PT_j) = S(e^{ik_{P_j}},e^{ik_{P_{j+1}}})A(P)$$
with $T_j$ the transposition operator of $j$ and $j+1$. Bethe's intuition is based on the idea that such many-body problems  should in the most, yet non-trivial, case only depend on the scattering of two particles. Such reduction of many-particle scattering is intertwined with the Yang-Baxter equation, 
\begin{eqnarray}
~~~~~~~~~~~~~~~~~~~~~~S(u_1,u_2)S(u_3,u_1)S(u_2,u_3)= S(u_3,u_2)S(u_1,u_3)S(u_2,u_1)
\label{eq:Yang-Baxter}
\end{eqnarray}

This framework turns out to be the most natural for solving the many-particle eigenstates.

First of all, one should wonder what the dynamic actually does. Lemma's (\ref{Lemma: one-particle}) offers us an elegant picture. Under the action of an interaction terms, a particle will hop to another size at a distance $<L$. After hopping, this particle is split as a new superposition of other particles. This transformation is determined by the Transfer matrix and the hopping distance.
Secondly, we ask how the interaction affects multiple particles simultaneously. When each particles are at the large distance, greater than $L$, from each other, this case never happens. This begs to questions of whether we should actually care about such case. We prove in Lemma (\ref{Lemma: Asymptotic Regime}), that, indeed, we should not do as such. For such asymptotic regime to valid, it is imperative that the main weight of the wave-function consists of particles far apart. This is an additional restriction on the allowed density state, meaning that it should be relatively low as it can be found in more details in the statement of the Lemma.

Which particles should we choose? In the one-particle section, we learned that the states are invariant under the gauge transformation (\ref{eq: main gauge}). 
The best choice turns out to be,

$$B^{(i)}_a=A^{(i)}X_a + e^{-i k_a}Y_aA^{(i)} - A^{(i)}Y_a,~Y_a =  (\operatorname{Id}-e^{-ik}\Gamma^*)^{-1}[X_a]$$
where $B^{(i)}_a$ should be interpreted as a particle with momentum $k_a$.

Generally, it turns out that we cannot find the whole emergent spectrum.  However, due to the asymptotic regime and the frustration freeness of the Hamiltonian, the symmetric states emerge as eigenstates. The coefficients $A(P)$ satisfy,
$$ A(T_{\alpha,\beta}P) = A(P)$$
for the transpositions $T_{\alpha,\beta}$. These states represent the particles as hopping bosons. 

As a consequence of the Bethe Ansatz, the energy of a m-particle state is the sum of the individual energies,
$$ E[X_{a_1},k_1;\ldots;X_{a_n},k_m]= \sum_{j}E[X_{a_j},k_j]$$
This implies the scaling of the energy with the particle number, $m  \times 2L$. Because in physics only differentials of the energy matters, a particle conservation is a consequence of this scaling.
 
 Strangely enough, we can solve the whole energy spectrum for normal unital quantum channels. Recalling our past discussion about the one-particle spectrum, we saw that there were different particle types for each momentum. The momentum was an additional label of the particle type. This is not the case for normal unital quantum channels. Thus there are at most $D^2-1$ different particles, including their anti-particles. An interesting consequence is that we are able to solve the low energy spectrum exactly. Utilizing the parity symmetry, the spectrum is thus simply any symmetric and antisymmetric superposition.
 Hence,
$$A(T_{\alpha,\beta}P)=\pm A(P) $$
  
 This complete results are stated formally in Theorem (\ref{Theorem: free Fermions}) and its corolary (\ref{Corollary: Exactly Solvable}) with their proofs in this appendix, which we summarize modestly as,
 \begin{quote}
 	\textit{ Every Injective Matrix Product State is the vacuum of a Hopping Theory}
 \end{quote}
\section{An Emerging Virtual State Space for Stochastics}
\label{section:Stochastics}
The main subject of this work is centralized around properties of some local Hamiltonian. Our tensorial notation of parent Hamiltonians (\ref{fig:parentH}) showed us that natural local sub-algebra's appear in the virtual space of the tensor network. Furthermore, these were related only and solely to the transfer matrix. We even showed that the low energy spectrum was exactly solvable, thereby interpreting some generators of these sub-matrix-algebras as physical particles.
This point of view offers particularly interesting research directions in the field of stochastics, quantum and classical Markov processes.
One could label each of these virtual particles and their position inside the network. As we showed earlier, under the action of the Hamiltonian, these particles are annihilated at the site of the action and new particles are created at the boundary of the interaction. If it turns out that such actions happens with a positive probability, then we can talk about a Markov process on this particular state space $\Omega$.
We briefly expose some applications for the Glauber dynamics, and leave the possibilities open for future work in this direction. 
\begin{figure}[t!]
	\hbox{\hspace{+3 cm}\includegraphics[width=0.7\textwidth]{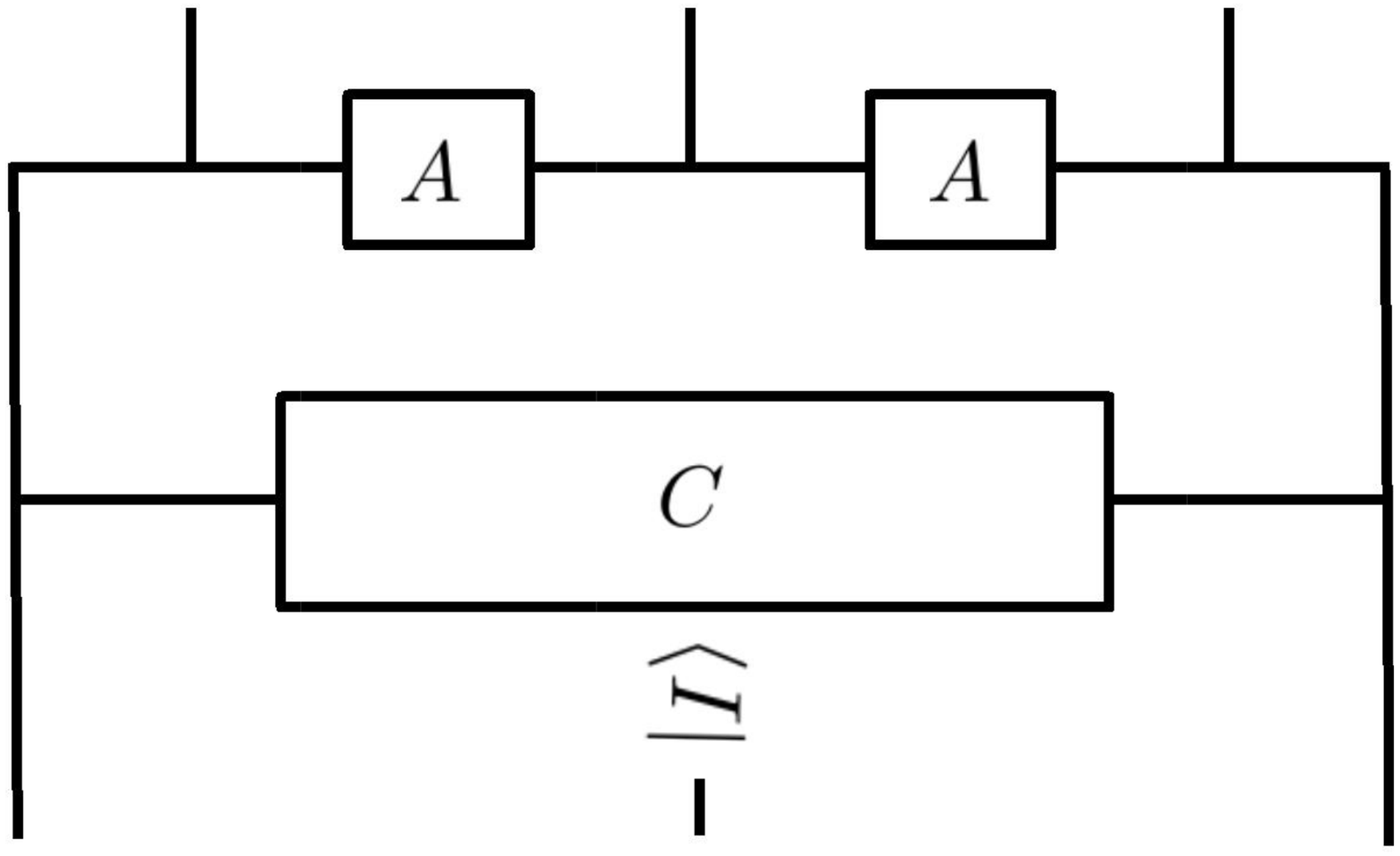}}
		\vspace{-2cm}
	\caption{A tensorial representation of the famous 1-dimensional Glauber dynamics often studied for Ising.}
	\label{fig:Glauber}
\end{figure}

In 1953, an algorithm was presented for calculating a particular canonical ensemble. This method of computation grew to be known as the Metropolis-Hastings \cite{chapter 2: Metropolis} algorithm. This algorithms is based on a Markov chains which takes random samples from a probability distribution. A famous use of this algorithms is for studying the gibbs state of Ising model,
$$\rho_{\beta} = \frac{e^{\beta H}}{\operatorname{Tr}e^{\beta H}},~~H=\sum_j s_{j}s_{j+1},~~s_j \in\{+1,-1\}$$
The Matrix Product State representation of the one-dimensional Ising model is given by the tensor, 
$$ A_{ab}^{(i)}= \delta_{i,a}\delta_{i,b}A_{a,b},~~A = \frac{1}{e^{\beta}+e^{-\beta}}\left(\begin{array}{cc}e^{\beta} & e^{-\beta}\\e^{-\beta} & e^{\beta}\end{array}\right)$$
The most simple Markov chain one can think of is the so-called Glauber dynamics in the state space of spin. At each time step, one flips a spin with a probability computed from the configuration of its neighbour. It should be obvious that the generators of such dynamics are local and actually related to a parent-Hamiltonian of some Matrix Product State. Therefore, the discussion of this work should offer some interesting new perspectives on this very well-known topic.
On can check that the local generator, $\mathcal{L}=\sum_{j}\mathcal{L}_{j-1,j,j+1}$, $\mathcal{L}_{j-1,j,j+1}= T_{j-1,j;j+1}-\mathbb{1}$ with $ T_{j-1,j;j+1}$ given in figure (\ref{fig:Glauber}), is the generator of such Glauber dynamics. In this case, $C$ is given by the following equation,
$$ C \circ A^2 = \mathbb{1},~~ (A\circ B)_{ij}=A_{ij}\dot B_{ij}$$
with the operation -$\circ$ being the Hadamard product.
In the spirit of the main discussion, one can find that Pauli matrix $Z$ is one-and only natural virtual particle. Therefore, we get a new state space by using the first quantization language with such particles. Writing the state, 

$$\phi(Z^{(a_1)},\dots,Z^{(a_{N})}) = \sum_{\vec{i}}\operatorname{Tr}\left(A^{i_1}Z^{(a_1)}\dots A^{i_N}Z^{(a_N)} \right)|\vec{i}\rangle,~~a_j \in\{0,1\}$$
similarly to the space of spin configurations $\sigma \in \{+1,-1\}^N$, we have the possibility to define a state space which mark the presence of a particle $Z$ at the various sites.
As  promised this new state space, is closed under the dynamics, as one can verify that,
$$T_{j-1,j,j+1}\phi(Z^{(a_1)},\dots,Z^{(a_{N})})=\sum_{b_{j-1},b_{j+1}} \tau^{a_{j-1},a_j,a_{j+1}}_{b_{j-1}b_{j+1}}\phi(Z^{(a_1)},\dots,Z^{(b_{j-1})},\mathbb{1},Z^{(b_{j+1})},\dots,Z^{(a_{N})})$$
and,
$$\tau^{a_{j-1},a_j,a_{j+1}}_{b_{j-1},b_{j+1}}=\frac{1}{4}\operatorname{Tr}\left(A^{T}Z^{a_j}A^{T}Z^{a_{j-1}+b_{j-1}}CZ^{a_{j+1}+b_{j+1}}\right)$$
One sees that the expression $A^{T}Z^{a_j}A^{T}$ is symmetric when $a_j=0$ and anti-symmetric for $a_j=-1$. This implies that there will be no creation of additional particle $Z$, but only a hopping to the left or the right with possible annihilation if another particle is already present,
$$\tau^{a_{j-1},1,a_{j+1}}_{a_{j-1}+1,a_{j+1}}= \tau^{a_{j-1},1,a_{j+1}}_{a_{j-1},a_{j+1}+1} $$ 
One dimensional Glauber dynamics are known to be exactly solvable. This small calculation yields us the famous relations for the dynamical correlation function,
$$C_{n-m}(t,s)=\langle s_n(t) s_m(s)\rangle,  \partial_t C_{n}(t,s)=-C_{n}(t,s)+\lambda (C_{n-1}(t,s)+C_{n+1}(t,s))$$

This first quantization, which appears naturally from this notation of the parent Hamiltonians (\ref{fig:parentH}),(\ref{fig:Glauber}), presents us an interesting new state space. The stochastic dynamics, of which the tensor network is a fixed point, acts an error corrector. The virtual particles are moved around, created and annihilated a certain rate, while the whole converges towards the vacuum, i.e. the Matrix Product State without any particles added to the virtual space. 
\section{Revisiting Localization}
\label{section:Localization}
In 1958, Anderson \cite{chapter 1: Anderson} presented the results of his study of diffusion in random lattices. He showed that for a particular form of a Hamiltonian consisting of random potentials, when the fluctuations were sufficiently large, diffusion would cease. The result first considered as a strange mathematical trickery, grew to a large field, and better understanding of insulation-conductor transition, such as Mott transition \cite{chapter 4: Tingelen}.

50 years later, a part of the research community focussed their attention on the so-called Tensor Networks  \cite{Fannes,chapter 2: Perez MPS rep}. While much richer phases were discovered beyond Landau, the philosophy of local order parameter were shifted towards variation of local tensors \cite{chapter 4: Gu}.
Furthermore, Gross-Pitaevski equations have been generalized for such states \cite{chapter 4: Haegeman QGross}. Time-evolution can be understood as a transport along a Tensor-Network manifold.

Even-though localization is a dynamical property, we show in this work that by combining both ideas presented above, the same phenomena is realised on the level of the manifold of Matrix Product States. The dynamic is exchanged with variation of the local tensors in the manifold, on top of which we consider some random fluctuations.

\begin{figure}[t!]
	\hbox{\hspace{+3 cm}\includegraphics[width=0.7\textwidth]{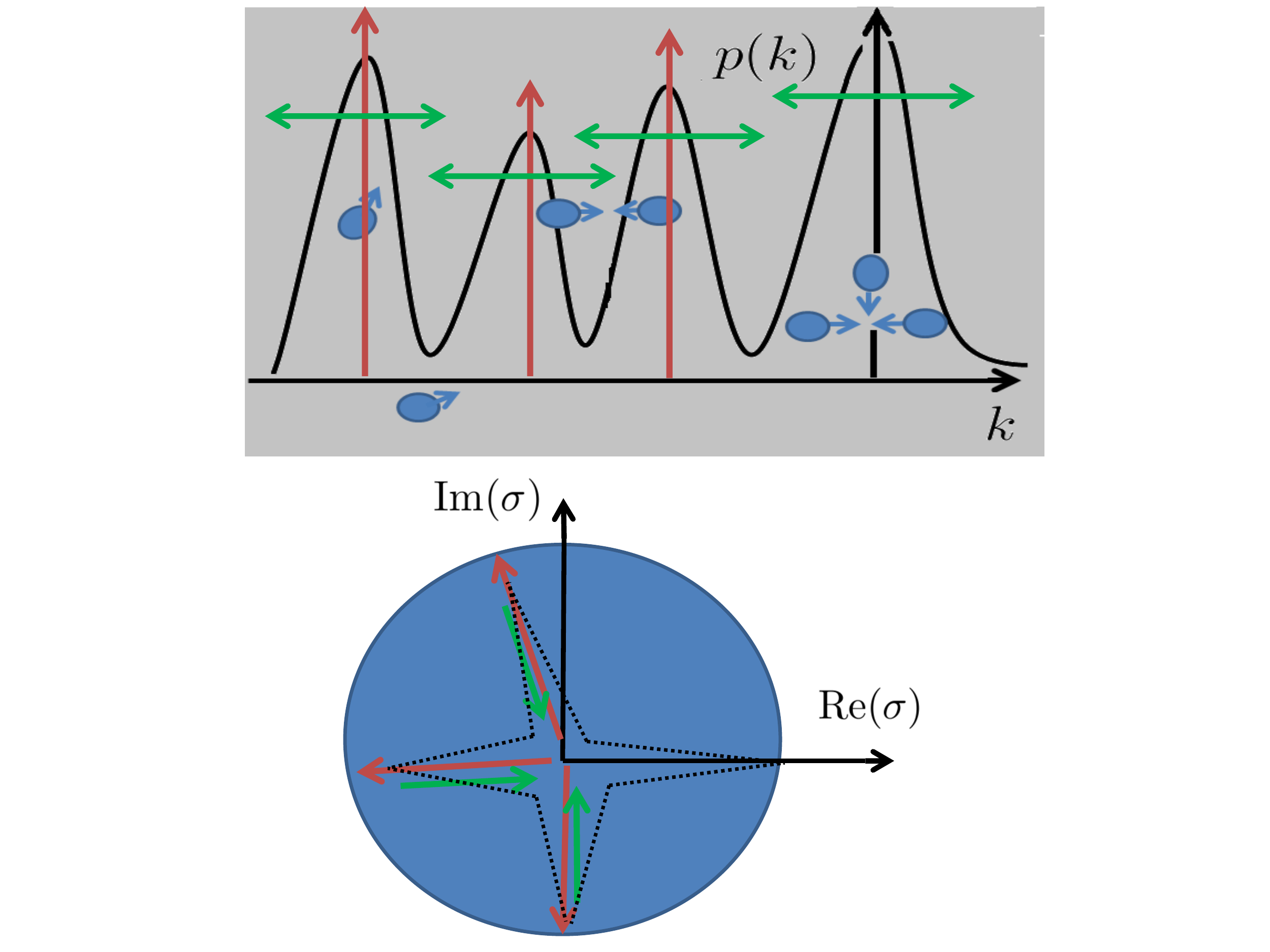}}
	\caption{$|\mbox{MPS}\rangle$ is a sea consisting of large number particles scattering with one-another. The distribution of the momentum has a finite number of peaks corresponding to the phase of the largest eigenvalues of the transfer-matrix. Diffusion (red) decreases the variance of the peaks, while localization increases it. On the level the eigenvalues of the transfer matrix, this corresponds to respectively pushing the eigenvalues towards unit circle and the center.}
	\label{sea}
	\vspace{+2mm}
\end{figure}
A very elegant intuitive interpretation of MPS leads us to this insight.
In the first part of this work, we showed that the Fourier Transform of the transfer matrix of a Matrix Product, contains information about low-energy particles of a local Hamiltonian.
Therefore, we propose that a translational invariant MPS can be established to be sea of scattering particles, from which the momentum-distribution has a finite number of peaks, as seen in figure (\ref{sea}),
\begin{eqnarray}
|\mbox{MPS}\rangle=\sum_{k,m}  |m(0)\rangle + |(m(k)+m(-k))\rangle +|(m(k_1)+m(k_2)+m(k_3))|k_1+k_2+k_3=0\rangle  +\dots
\label{seaMPS}
\end{eqnarray}

This interpretation of MPS becomes more transparent, when trying to explain the transition from diffusion to localization as we do in section (\ref{sect:loc}).
As shown in \cite{chapter 4: Zauner Excitations}, by studying the z-transform of 2-point-correlation functions we can probe the weight of various momenta. In some sense, it means that the phase of an eigenvalue is related to the momentum of some set of particles in the sea.
We also argued earlier that the momentum densities $p(k)$ are related to the eigenvalues of the transfer matrix of the MPS, and this was exact for normal transfer matrices.\newline

\begin{figure}[t!]
	\vspace{-5mm}
	\hbox{\hspace{+2.5 cm}\includegraphics[width=0.7\textwidth]{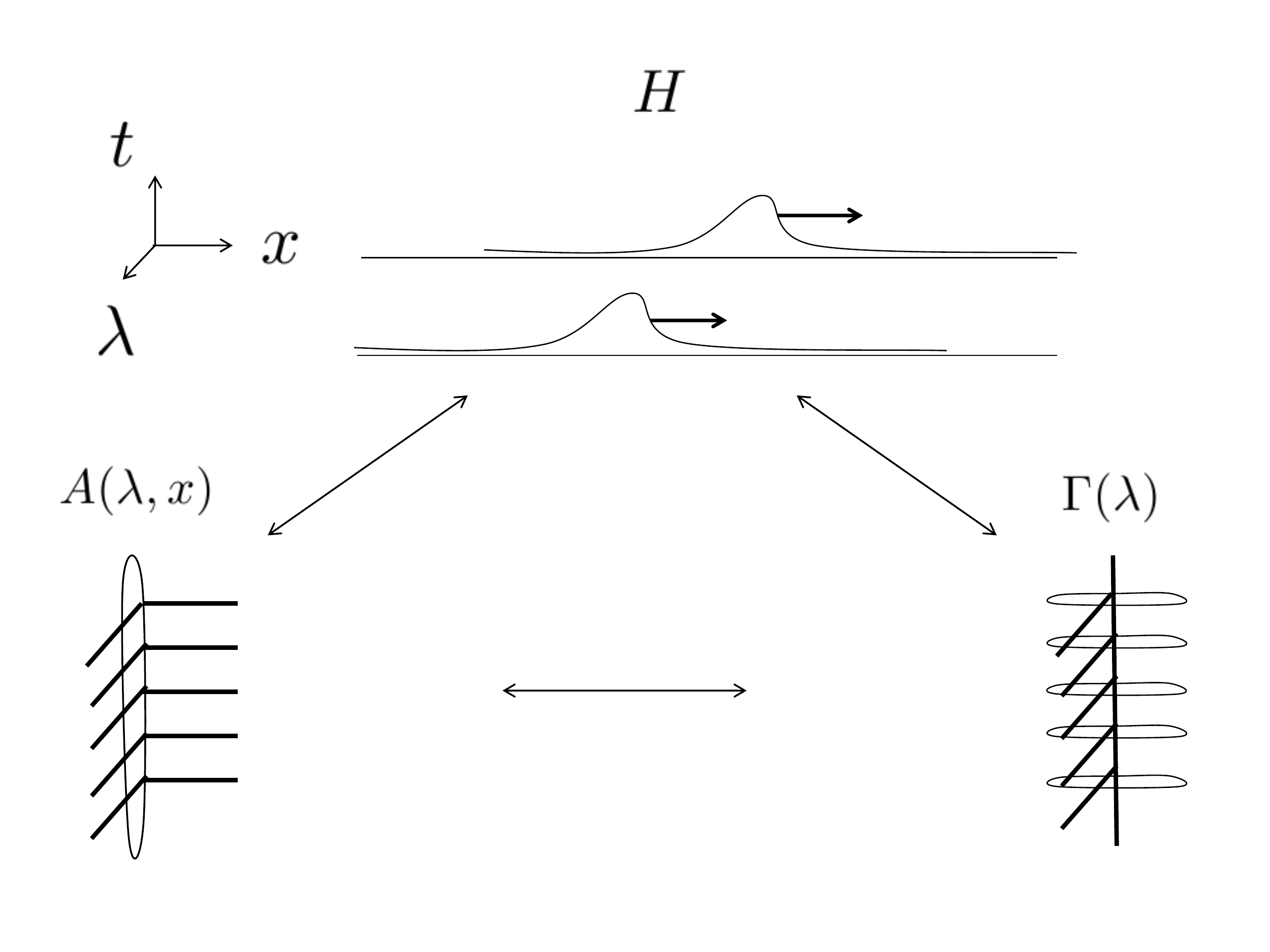}}
	\caption{Simple illustration of the different mapping $1+1+0 \leftrightarrow 0+1+1\leftrightarrow 1+0+1$.}
	\label{holo1d}
	\vspace{+2mm}
\end{figure}

\section{Set-up}
\label{sect:set-up}
Consider a continuous one-parameter family of Matrix Product States, $\{B^{i}[s]=A^{(i)}[s](\lambda)+W^{(i)}[s]\}$,
\begin{eqnarray}~~~~~~~~~~~~~~~~~~~~~~|\psi\{B^{i}[s]\}\rangle =\sum_{\vec{i}}\operatorname{Tr}\left(B^{i_1}[1]\dots B^{i_1}[N]\right)\sigma^{+i_1}\dots \sigma^{+i_N}|\Omega\rangle
\label{bosonMPS}
\end{eqnarray}
As presented earlier the parameter $\lambda \in I$ is the equivalent of the time-evolution, in the sense of quasi-adiabatic evolution \cite{chapter 1: Osborne quasAdev}. The extra terms in the tensors $W^{(i)}[s]$ are taken randomly from some measure $d\mu(W^{(i)}[s])$.

We need now to make the claimed equivalence more precise. Let us start with the time evolution of the density.
The creation of a particle at site $0$, followed by a measurement the outcome at site $n$, can be calculated from the correlation function. As we see below this can be related to the transfer matrix of the MPS-family $\Gamma(\lambda)=E_\mu\left(\sum_j B^{j}[s]\otimes \overline{B}^{(j)}[s]\right)(\lambda)$,
\begin{eqnarray*}
|\psi(0,n)|^2&\propto|\langle O_0 M_n\rangle-\langle O_0 \rangle \langle M_n\rangle|\\
&\propto \frac{\left|\langle l[O]|\left[\Gamma^n-\Gamma^\infty\right]|r[M]\rangle\right|}{\sqrt{\langle l[O]|l[O]\rangle}\sqrt{\langle r[M]|r[M]\rangle}}\\
&\leq \|\Gamma^n(\lambda)-\Gamma^\infty(\lambda)\|_1
\end{eqnarray*}
Anderson studied the time-evolution of the probability distribution, $|\psi(0,n,t)|^2$ of a particle initially at the origin and evolving under a dynamic described by a Hamiltonian $H=\tilde{H}+\sum_xW(x)$. He,\cite{chapter 1: Anderson,chapter 4: Tingelen}, showed the amazing result that in one and two dimension(s) for finite magnitude of fluctuations, such distribution remains localized in the sense that,
$$\sum_n |n|\ |\psi(0,n,\lambda)|^2 < C$$
As argued above, we can now argue localization when studying the function $\Xi(\lambda,W)$,
$$\sum_n |n|\ |\psi(0,n,\lambda)|^2 \leftrightarrow \Xi(\lambda)=\frac{\sum_n|n|\  \|\Gamma^n(\lambda)-\Gamma^\infty(\lambda)\|_1 }{\sum_n\  \|\Gamma^n(\lambda)-\Gamma^\infty(\lambda)\|_1} $$
for some measure of magnitude of the fluctuations $W$.

The next step, is for us to properly connect $\lambda$ with time $t$ and the diffusion. While $\lambda$ maybe defined on a small interval, time always goes to infinity. Therefore, in the same spirit as the quasi-adiabatic evolution \cite{chapter 1: Osborne quasAdev}, time must be varied much more slowly when varying $\lambda(t)$. Additionally, the precise relation $\lambda(t)$ is fixed by the property of diffusion of the family of MPS,
$$\Xi(\lambda(t),M=0)\leq D \sqrt{t}$$
The set-up is summarized in the table below,
\begin{center}
	\begin{tabular}{ | l | c | }
		\hline
		Dynamical & Tensor Network\\ \hline \hline
		$W(x)$ & $W^{(i)}[x]$\\\hline
		$t$ & $\lambda(t):\Xi(\lambda(t),W=0)\leq D\sqrt{t}$ \\ \hline
		$U(t)=\exp(itH)$ & $A^{(i)}(\lambda)+W^{(i)}[x]$\\\hline
		& \\
		$|\psi(0,n,t)|^2$ & $\frac{  \|\Gamma^n(\lambda)-\Gamma^\infty(\lambda)\|_1 }{\sum_n\  \|\Gamma^n(\lambda)-\Gamma^\infty(\lambda)\|_1}$ \\
		\hline
	\end{tabular}
\end{center}
\section{Localization}
\label{sect:loc}
Finally, we can proceed and describe localization.
While proving  Anderson's  localization is extremely challenging, this mapping onto a Tensor Network-picture gives a fresh simplistic view of the problem.
The following explanation is illustrated in figures (\ref{sea}), and (\ref{holo1d}).
Let us bring our insight of the translational MPS (\ref{seaMPS}) come to play. 
The reason for diffusion in the first place is that the leading terms of the wave function is some wave-packet. This is compatible with our picture as this means that the distribution $p(k)$ of momentum has sharp peaks around a finite number of momenta. 
Localization arises, when leading terms of the wave-function consists of a large number of different wave-packets with different momentum. This is again compatible with the sea-picture (\ref{seaMPS}). A translational invariant MPS has zero-momentum. In the case of a sea of localized states, the distribution $p(k)$ becomes more and more peaked around the origin. Therefore no particles are scattering as, so they are all localized.

Let us translate the sea-picture, 0+1+1, for the transfer-matrix, i.e. 1+0+1. 
We saw in the previous section that the probability distribution $|\psi(0,n,t)|^2$ is related to the spectrum of the transfer matrix $\Gamma$. 

\begin{figure}[t!]
	\hbox{\hspace{+2 cm}\includegraphics[width=0.8\textwidth]{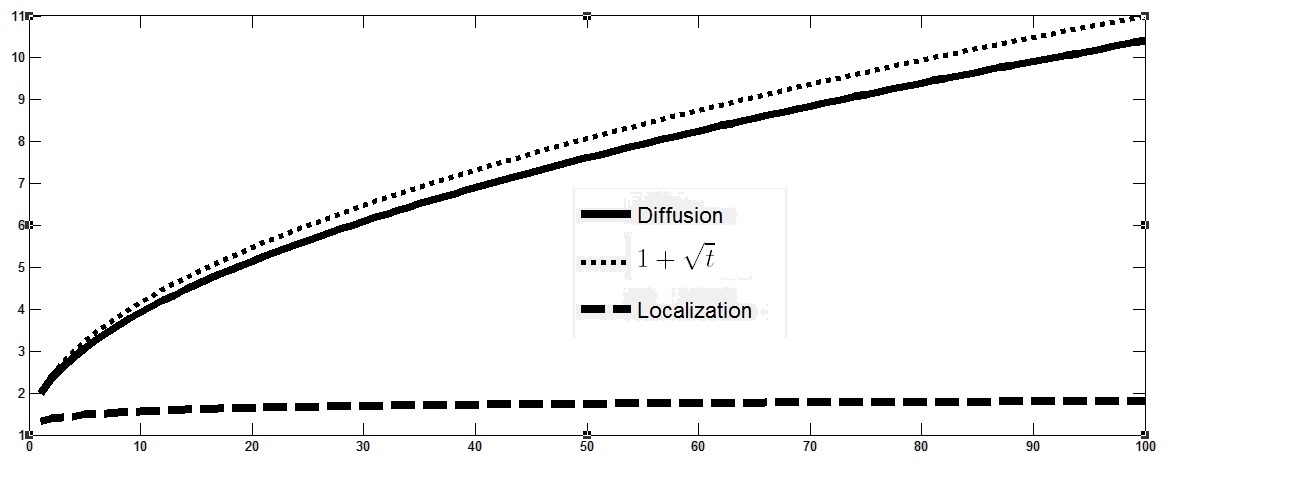}}
	\caption{Plot of $\Xi(\lambda)$  vs $t$ for the example (\ref{exampleAKLT}) with $N=100$ and $W=1$.}
	\label{diffloc}
	\vspace{-5mm}
\end{figure}

The probability distribution $p(k)$ can in principle be studied through by doing some tomography and next studying the $z$-transform of the 2-point correlation function. Keeping this idea in mind, it is acceptable to see the phase of the eigenvalues of the transfer matrix, as the momentum of a set of particles in the MPS-sea.
Clearly then diffusion arises when the eigenvalues are close to the unit circle. Next, local fluctuations increase the variance of the peaks. This is, again, can be understood from the spectrum, as this means that the eigenvalues are moved towards the origin.
The spectrum of generic transfer matrices have often a star-shape. This is mathematically understood for completely positive operators, as the eigenvalues on the unit circle are a representation of some abelian group \cite{Fannes,chapter 2: Perez MPS rep}. This fits perfectly with the interpretation of the MPS as a stationary sea. The total momentum is zero, and so the sum of the phase must be zero, as it is for any non-trivial one-dimensional representations of abelian groups.
In the extreme case of localization all eigenvalues are zero except one. remaining eigenvalue, which is 1, is the zero-momentum of the sea.

Let us illustrate this with an example,

\begin{example}
	\label{exampleAKLT}
	Consider the local tensors, 
	\begin{eqnarray*}
	A^{(-1)}(\lambda)+W^{(-1)}=\left(\sqrt{1-\lambda}+w_1/\sqrt{2}\right)\sigma^+,\\
	A^{(+1)}(\lambda)+W^{(+1)}=\left(\sqrt{1-\lambda}+w_1/\sqrt{2}\right)\sigma^-\\
	A^{(0)}(\lambda)+W^{(0)}=\left(\sqrt{\lambda}+w_1/2\right)\sigma^z+w_2/2 \mathbb{1}
	\end{eqnarray*}
	with $w_j \sim N(0,W)$ and mutually independent for some finite variance $W$. Here we take $\lambda\in [1/2,1]$, and see that AKLT belong to the family for $\lambda =2/3$. 
	By choosing the right function $\lambda(t):[1,\infty]\to [1/2,1]:t\to \frac{\sqrt{t}}{1+\sqrt{t}}$, we can plot $\Xi(\lambda,M)$ in figure (\ref{diffloc}).
	A simple calculation can verify, 
	$$\Xi(\lambda(t),W=0)\leq C'(1+\sqrt{t}),~~\Xi(\lambda(t),W=1)\leq C $$
\end{example}
As promised, we see in figure (\ref{diffloc}) that our choice $\lambda(t)$ yields diffusion. Similarly to Anderson's result, a finite magnitude of the fluctuation leads to localization.

As mentioned earlier, localization is especially interesting for studying mott insulators. Hence, we need to be able to extend the result to fermions.
This can of course be done for Fermionic MPS \cite{chapter 4: Kraus Fermion PEPS}, by replacing $\sigma^{+i_k}\leftrightarrow\psi^{\dagger i_k}$ with fermionic creation operators. The transfer matrix may vary depending on the local observable, $\Gamma=A^{0}\otimes \overline{A}^{0}+A^{1}\otimes \overline{A}^{1}$ or $\Gamma=A^{0}\otimes \overline{A}^{0}-A^{1}\otimes \overline{A}^{1}$. The rest of the set-up is then the same, this time however, we need to choose transfer matrix with the largest $\Xi(\lambda,M)$.

\section{Conclusion}
In this work, we proved that all injective Matrix Product States are the unique ground-states of hopping models.
 This was done  by studying the spectrum of parent Hamiltonians in the long range and system size limit, and eventual renormalization of the transfer operator of the Matrix Product ground state, under the restriction that an asymptotic regime remains valid for low particle densities 
We showed that, in this limit, the spectrum depends solely on the properties of a quantum channel, i.e. the transfer matrix. The derived formula for the low energy spectrum was very reminiscent of the K\"all\'en -Lehmann spectral representation in quantum field theory.
Particles types are represented by eigenvectors of the fourier transform of the transfer matrix (\ref{eq:FourierTransfer}). Higher energy wavefunction are constructed by putting these particles in the virtual space of the ground-state with different momenta. 

Secondly, we consider a renormalization of the spectrum of the transfer matrix with range $L$. The choice was made in such a way that the ground state becomes "critical". It results that the particles derived earlier are not necessarily stable.
The first quantization breaks down when a particle can be fused from two others, of which the respective eigenvalues are close to the unit circle.

This work introduces an interesting interpretation of  translational invariant Matrix Product States as nothing but stationary sea of particles scattering with one-another. This insight motivated us to understand Anderson localization on the level of Tensor Networks. The property is introduced to the Matrix Product State manifold without considering any Hamiltonian-dynamic. The local random potentials are exchanged with fluctuations of the local tensor of the network. The time-evolution is reinterpreted as transport along the manifold which moves the eigenvalues of the transfer towards the unit circle. This is shown under proper rescaling of the parameter with $t$ to be equivalent to diffusion. Finally, it is illustrated how finite magnitude of fluctuation kills the diffusion similarly to the result proven by Anderson.

\section*{Acknowledgements}
We thank Jutho Haegeman for helpful discussions.
We acknowledge financial support by the FWF project CoQuS No.  W1
210N1 and project QUTE No. H20ERC2015000801.
\newpage
\appendix
\section{Statements and Proofs of Lemma's and Theorem}
\label{sec: tech}
The technical details of the limit are presented in this section. We start from a one-dimensional spin chain of size $N$ with interactions of range $L$.
We derive a few convergence rates for $L\to \infty$ for the low particle number eigenstates and their energy under a possible choice of re-normalisation of the transfer matrix $\Gamma(L)$.

We need to compute the the action of the Hamiltonian onto a m-particle subspace. As, we see a very natural new basis appears in the virtual space. However, this basis has to be orthogonalized. We, therefore, will have to get some information about the Gram-matrix of each subspace.\newline
First, we start our discussion with results about the one-particle states. In Lemma (\ref{Lemma: one-particle}) we demonstrate that the particle have a simple finite-dimensional representation closely related to a quantum channel.
Next, when combining these particles to construct multi-particle states, it can be assumed that the particles are very far from each other. Since the dynamic is frustration free, this implies that the theory is non-interactive, and the total energy is the sum of the energies of the  individual particles.

\begin{Lemma}[One-particle energy]
\label{Lemma: one-particle}	
Given a Matrix Product State $|\phi\rangle$, the one-particle states with momentum $k$, $|\psi\{X_a,k\}\rangle$, given by, 
\begin{eqnarray}
\label{eq:one-particle}
|\psi\{X_a,k\} =\sum_n e^{ikn)}|\phi\{X_a,n_a\}\rangle \nonumber\\ 
|\phi\{X_a,n\}\rangle=\sum_{\vec{i}}\operatorname{Tr}\left(A^{i_1}\dots A^{i_{n-1}} A^{i_{n}}X_a A^{i_{n+1}}\dots A^{i_N}\right)|\vec{i}\rangle
\end{eqnarray}
are eigenvectors of the local Hamiltonian $H(L)$, (\ref{eq:parentH}) with as choice of for the tweaking matrix $C$ the relation (\ref{eq:Correction}),
$$\frac{\|H|\psi[X_a,k]\rangle -E[X_a,k]|\psi[[X_a],k]\rangle\|_2^2 }{\langle\psi[\overline{X_a},k]|\psi[X_a,k]\rangle} \to 0 $$
The values  $E[X_a,k]$ are the solutions of the finite-dimensional eigenvalue problem of the matrix,
$$\left(h_{(ab)}\right)=\frac{\operatorname{Tr}\left(\rho^{-1}\mathcal{T}_k[\tilde{X}_a^\dagger] \mathcal{T}_k[\tilde{X}_b]\right)}{\operatorname{Tr}\left(\tilde{X}_a^\dagger \mathcal{T}_k[\tilde{X}_a]\right)},X_a =\sum_b U_{a,b}\tilde{X}_b,~ U^\dagger U = \mathbb{1}$$
The basis vectors $\tilde{X}_a$ are the eigenvalues of the Fourier transform of the transfer matrix (\ref{eq:FourierTransfer}) ,
$$\mathcal{T}_k[\tilde{X}_a]= \lambda_a X_a$$
For momentum $k=0$, there is the additional constraint of orthogonality onto the vacuum,
$$Q^*_{\rho}\circ\mathcal{T}_k\circ Q_{\rho}[\tilde{X}_a]= \lambda_a X_a,~Q_{\rho}[.]=[.]-\rho$$
\end{Lemma}
\begin{proof}
	By our choices given by the equations (\ref{eq:parentH}) and (\ref{eq:Correction}), the Hamiltionan $H(L)$ is local and frustration free.
	Let us fix the momentum $k$.
	Our first step is to show that the one-particle subspace spanned by states of the type $|\psi\{X_a,k\}\rangle$ is closed under the action of $H(L)$, in the limit $L\to \infty$.
	Knowing this, we can proceed and find an orthogonal basis for this subspace, and compute the effective Hamiltonian. 
	
	A look at the choice of corrections for $H(L)$ to be frustration free, yields that in the limit $L\to \infty$ all these corrections should converge,
	$$ C \to \rho^{-1}\otimes \mathbb{1},~ L\to \infty$$
	Using the frustration freeness of the Hamiltonia we can bound this first error,
	
	$$\Big|\Big|H|\psi\{X_a,k\}\rangle-\sum_{j, j\leq n \leq j+L}e^{ikn} \tilde{H}_{j,j+L}|\phi\{X_a,n\}\Big|\Big|_2 \leq \epsilon_1(L)$$
		where the correction term of the interactions terms $\tilde{H}_{j,j+L}$ are given by $\rho^{-1}\otimes \mathbb{1}$. This error typically decays as $O(D^2 |\lambda|^L)$, with $\lambda$ the second largest eigenvalue of the transfer matrix $\Gamma$.
		
	This simplifications leads to the actions,
	\begin{eqnarray}
	\label{eq:one-particle action}
	\sum_{j, j\leq n \leq j+L}e^{ikn} \tilde{H}_{j,j+L}|\phi\{X_a,n\}\rangle =2L|\psi\{X_a,k\}\rangle-2|\psi\{\mathcal{H}[X_a],k\}\rangle \nonumber \\  -\sum_{j,\alpha,\beta}e^{ik j}|\phi\{e_{\alpha},j,e_{\beta},j+L\}\rangle \operatorname{Tr}\left(e_{\alpha}^\dagger \mathcal{C}\{X_a,k\}[e_{\beta}]\right)
	\end{eqnarray}
	with, 
	\begin{eqnarray*}
	\mathcal{H}_k[.]=R_{\rho^{-1}}\circ\mathcal{T}^{(L)}_k[.],\\
	\mathcal{C}\{X_a,k\}[.]=\sum_{j=1}^{L-1} e{ikj}R_{\rho^{-1}}\circ \Gamma^{j}\circ  \tilde{L}_{X_a}\circ \Gamma^{L-j} [.],\\
	 R_{\rho^{-1}}[.]=[.]\rho^{-1},~Q_{\rho}[.]=[.]-\rho \operatorname{Tr}[.],~~\tilde{L}_{X}[.]=Q_\rho[X_a Q_\rho[.]]
	\end{eqnarray*}
	An important super-operator is the Fourier Transform which is here approximated by,
	
	$$\mathcal{T}^{(L)}_k[\Gamma]=R_{\rho}+\frac{1}{2}\left( \sum_{n=1}^{L-1}  \Gamma^n\circ  \mathcal{R}_{\rho}\ e^{ikn}
	+\mathcal{R}_{\rho}\circ \Gamma^{* n}\ e^{-ikn}\right)$$
	The first line in the equation (\ref{eq:one-particle action}), is an effective actions on the finite-dimensional subspace spanned by $X_a$.
	However, the states $|\psi\{X_a,k\}\rangle$ are not necessarily orthonormal. This is sorted out by computing the Gram-matrix of the $1-$particle subspace,
	$$G^{(1)}_k(X_a,X_b)= N\langle\psi[\overline{X_a},k]|\psi[X_b,k]\rangle = \operatorname{Tr}\left(X_a^\dagger	\mathcal{T}^{(N)}_k[\Gamma][X_b]\right)$$
	The super-operator $\mathcal{T}_k[\Gamma]$ is hermitian and has thus orthonormal eigenvectors. Using the eigenvectors, we can construct the effective one-particle Hamiltonian presented in the statement of the Lemma.
	One should note that the identity particle $X_a=\mathbb{1}$ with momentum $k=0$ is nothing but the ground-state. However, the particle is not an eigenvector for $k\not=0$. As all one-particle states need to be orthogonal onto the vacuum state, for $k=0$, the additional constraint for $k=0$ has to be imposed,
	$$\langle \psi|\psi[X_b,k]\rangle =\operatorname{Tr}\left(\rho X_b\right)=0 $$
	
	We should finally discuss the error of the action,
	
	\begin{eqnarray*}
	\frac{\|H|\psi[X_a,k]\rangle -\lambda|\psi[[X_a],k]\rangle\|_2^2 }{\langle\psi[\overline{X_a},k]|\psi[X_a,k]\rangle}
	\\
	= \frac{N}{N \operatorname{Tr}\left(X_a^\dagger	\mathcal{T}^{(N)}_k[\Gamma][X_a]\right) } \sum_{\alpha_1,\alpha_2,\beta_1,\beta_2,j} \langle\phi[e_{\alpha_1},0,e_{\beta_1},L]|\phi[e_{\alpha_2},j,e_{\beta_2},j+L]\rangle \\
	\overline{\operatorname{Tr}\left(e_{\alpha_1}^\dagger \mathcal{C}\{X_a,k\}[e_{\beta_1}]\right)} \operatorname{Tr}\left(e_{\alpha_2}^\dagger \mathcal{C}\{X_a,k\}[e_{\beta_2}]\right)
	\end{eqnarray*}
	
	Typically, this error dies off as $O(D^2|\lambda|^L)$. As all the errors are being reduced exponentially with $L$ without any dependence on the system size, the claim of our lemma is proven.
	  
\end{proof}
In the theory of scattering one compares the change of the wavefunction of two particles initially infinitely far from each other, with the resulting outcome after infinite time when they are again far apart. The total energy of the system, since both particles are uncorrelated in both cases, is the sum of the individual energy. 
This additive feature of the energy appears regularly in exactly solvable models or even integrable ones. The particles simply fill in some Fermi sea.

We prove in the next lemma, that this is also exactly what happens for parent Hamiltonians in the long range and large system size limit. The intuition is that the action of the Hamiltonian onto a m-particle wavefunction, approximately only evaluates the energy when the particles are very far apart. 
Clearly for such asymptotic regime to be valid, particles need enough empty space. Thus, there is a restriction on the particle number which scales with $L$.
\begin{Lemma}[Asymptotic Regime at Low Densities]
\label{Lemma: Asymptotic Regime}	
For a m-particle state with particles satisfying the conditions of Lemma (\ref{Lemma: one-particle}),
\begin{eqnarray}
\label{eq:m-particle}
|\psi\{X_\alpha,k_\alpha\}_{\alpha=1}^m\rangle =\sum_{n_1<\ldots<n_m} e^{i\vec{k}.\vec{n}}|\phi\{X_\alpha,n_\alpha\}\rangle \nonumber\\ 
|\phi\{X_\alpha,n_\alpha\}_{\alpha=1}^m\rangle=\sum_{\vec{i}}\operatorname{Tr}\left(A^{i_1}\ldots A^{i_{n_1}}X_1A^{i_{n_1+1}}\ldots A^{i_{n_m}}X_1A^{i_{n_m+1}}\ldots A^{i_N}\right)|\vec{i}\rangle
\end{eqnarray}
the action of the Hamiltonian takes into consideration the individual particles independently from the others,
\begin{eqnarray}
H|\psi\{B^{(i)}_\alpha,k_\alpha\}_{\alpha=1}^m\rangle &\approx 2m L |\psi\{B^{(i)}_\alpha,k_\alpha\}_{\alpha=1}^m\rangle + 2  |\psi\{\mathcal{H}[X_1],k_1;X_2,k_2;\ldots;X_m,k_m\}\rangle \nonumber\\
&+ \ldots+ |\psi\{X_1,k_1;\ldots;X_{m-1},k_{m-1};\mathcal{H}[X_m],k_m\}\rangle \nonumber
\end{eqnarray}
The error scales typically as $O\left(\sum_{n=1}^m(m-n+1)\left(\begin{array}{c}L\\ n\end{array}\right)N^{-n}D^{2(n+1)}\right)$	
\end{Lemma}
\begin{proof}
	The proof is more cumbersome in its notation than complex. As a first illustration, let us first look at what happens for 2 particles-basis states.
	One can verify,
	\begin{eqnarray}
	& H|\psi[X_{a_1},k_1;X_{a_2},k_2]\rangle=\left(4L|\psi[X_{a_1},k_1;X_{a_2},k_2]\rangle\right.\label{eq: 2-part LINE1}\\
	&\left.-|\psi[\mathcal{H}_{k_1}[X_{a_1}],k_1;X_{a_2},k_2]\rangle-|\psi[X_{a_1},k_1;\mathcal{H}_{k_2}[X_{a_2}],k_2]\rangle\right)\label{eq: 2-part LINE2}\\
	&-|\psi[\tau_{k_1}[X_{a_1}]X_{a_2},k]\rangle-|\psi[X_{a_1}\tau_{k_2}^*[X_{a_2}],k]\rangle \label{eq: 2-part LINE3}\\
	&-\sum_{\alpha,\beta}\sum_{0\leq j_1<j_2,|j_1-j_2|> L}e^{i\left[k_1 j_1+k_2 j_2\right]}|\phi[e_{\alpha},j_1;e_{\beta},j_1+L;X_{a_2},j_2]\rangle \operatorname{Tr}\left(e_{\alpha}^\dagger \mathcal{C}\{X_{a_1},k_1\}[e_{\beta}]\right)\label{eq: 2-part LINE4}\\
	&-\sum_{\alpha,\beta}\sum_{0\leq j_1<j_2,|j_1-j_2|> L}e^{i\left[k_1 j_1+k_2 j_2\right]}|\phi[X_{a_1},j_1;e_{\alpha},j_2,e_{\beta},j_2+L]\rangle \operatorname{Tr}\left(e_{\alpha}^\dagger \mathcal{C}\{X_{a_2},k_2\}[e_{\beta}]\right)\label{eq: 2-part LINE5}\\
	&-\sum_{j,\alpha,\beta}e^{i(k_1+k_2) j}|\phi[e_{\alpha},j,e_{\beta},j+L]\rangle \operatorname{Tr}\left(e_{\alpha}^\dagger \mathcal{C}^{(2)}\{X_{a_1},k_1;X_{a_2},k_2\}[e_{\beta}]\right)+\epsilon_*\label{eq: 2-part LINE6}
	\end{eqnarray}
	with, 
	\begin{eqnarray*}
	&\mathcal{C}^{(2)}\{X_{a_1},k_1;X_{a_2},k_2\}[.]=\sum_{0< j_1<j_2\leq L}e^{i \left[k_1 j_1+k_2 j_2 \right]}R_{\rho^{-1}}\circ \Gamma^{j_1}\circ \tilde{L}_{X_{a_1}}\circ \Gamma^{j_2-j_1}\circ \tilde{L}_{X_{a_2}}[.], \\&\tau_k[.]=\sum_{j=1}^Le^{-ikj}\Gamma^j[.]
	\end{eqnarray*}
	The lines (\ref{eq: 2-part LINE1}-\ref{eq: 2-part LINE2}) are the sought results for the 2-particles. The line (\ref{eq: 2-part LINE3}) results from the one-particles being transported due to the interactions and fusing with the other particle. Lines (\ref{eq: 2-part LINE4}-\ref{eq: 2-part LINE5}) are the errors of the one-particle approximation, we argued in the previous lemma. The last term (\ref{eq: 2-part LINE6}) results from the simultaneous action of the interaction on both particles. We have already demonstrated that  (\ref{eq: 2-part LINE4}- and (\ref{eq: 2-part LINE5}) decay exponentially with $L$. Hence, only (\ref{eq: 2-part LINE3}) and (\ref{eq: 2-part LINE6}) are new.
	One should notice that when normalized, the overlap of $m_1$ with $m_2$ particle states scales as,
	$$ |\langle\psi\{X_{\alpha},k_\alpha\}_{\alpha=1}^{m_1}|\psi\{X_\alpha,k_\alpha\}_{\alpha=1}^{m_2}\rangle| = O(N^{\frac{n_1+n_2}{2}})$$
	This is elaborated in more details in the proof of the theorem.
	The result is similar for the state in (\ref{eq: 2-part LINE6}) and taking it to be similar to a one-particle state.
	Looking specifically at each line. There are at most two new 1-particle states in the second line. Each new one-particle state is the result of $L D^2$ other terms.
	The argumentation is similar for (\ref{eq: 2-part LINE6}). There are $D^{6}L(L-1)$ contribution to this one-particle state.
	We omitted the error terms $\epsilon_*$ which come from the local identity operation of the interactions.
	
	$$\epsilon*= \sum_{n_1} 2(|n_1-n_2|-L)e^{i(k_1n_1+k_2n_2)}\sum_{1\leq n_2-n_1 \leq L } |\phi\{X_{a_1},n_1;X_{a_2},n_2\}$$
	Under normalization with $\langle\psi[X_{a_1},k_1;X_{a_2},k_2]|\psi[X_{a_1},k_1;X_{a_2},k_2]\rangle$, this error is of order $O(D^4L^2/N)$.
	
	With this in mind, we can tackle the m-particle case and regroup the different error types together,
	\begin{eqnarray}
&\Big|\Big|H|\psi\{B^{(i)}_\alpha,k_\alpha\}_{\alpha=1}^m\rangle -\left( 2m L |\psi\{B^{(i)}_\alpha,k_\alpha\}_{\alpha=1}^m\rangle + 2  |\psi\{\mathcal{H}[X_1],k_1;X_2,k_2;\ldots;X_m,k_m\}\rangle \right.\nonumber\\
&\left.+\ldots + |\psi\{X_{1},k_1;\ldots;X_{m-1},k_{m-1};\mathcal{H}[X_m],k_m\}\rangle \right)\Big|\Big|_2\nonumber\\
&\leq \| |\phi\{\mathcal{X}_{1,1}\{k_1\},X_2,k_2;\ldots;X_m,k_m\}\rangle+\ldots+|\psi\{X_2,k_2;\ldots;X_{m-1},k_{m-1};\mathcal{X}_{1,m}\{k_m\}X_1,k_1\}\rangle \|\nonumber\\
 &+ \| |\zeta\{\mathcal{X}_{2,1}\{k_1,k_2\};X_3,k_3;\ldots;X_m,k_m\}\rangle+\ldots\nonumber\\
 &+|\zeta\{X_2,k_2;\ldots;X_{m-2},k_{m-2};\mathcal{X}_{2,m-1}\{k_{m-1},k_m\};X_1,k_1;\}\rangle \|\nonumber\\
 &+\ldots+ \| |\zeta\{\mathcal{X}_{m-1,1}\{k_1,k_2,..,k_{m-1}\},X_m k_m\}\rangle+\ldots+|\psi\{\mathcal{X}_{m-2,2}\{k_{2},..,k_m\};X_1,k_1\}\rangle \|\nonumber\\
  &+ \| |\zeta\{\mathcal{X}_{m,1}\{k_1,k_2,..,k_{m}\}\rangle \| \nonumber +\epsilon_*
	\end{eqnarray}
	The operator $\mathcal{X}_{n_1,n_2}$ consists of at most $\left(\begin{array}{c}L\\ n\end{array}\right)D^{2(n+1)}$ terms. Under the states  $|\zeta\{\mathcal{X}\{q_1,\ldots,q_{m_1}\};X_1,k_1;\ldots;X_{m_2},k_{m_2}\}\rangle$, one should understand,
	\begin{eqnarray}
|\zeta\{\mathcal{X};X_1,k_1;\ldots;X_n,k_n\}\rangle \nonumber \\
 = \sum_{n_1+L\leq\ldots\leq n_{m_1+1},\alpha,\beta}e^{i(q_1+\ldots+q_{m_1}) n_1}|\phi[e_{\alpha},e_{\beta},n_1+L,X_1,n_2,\ldots,X_{m_2},n_{m_2+1}]\rangle \operatorname{Tr}\left(e_{\alpha}^\dagger \hat{\mathcal{X}}[e_{\beta}]\right) \nonumber
   \end{eqnarray}
with bounded super-operator $\hat{\mathcal{X}}$.
Again the error $\epsilon_*$ is shown to be,
	$$\epsilon*= \sum_{n_1} (|n_1-n_2|+|n_{m}-n_{m-1}|-2mL)e^{i(k_1n_1+k_2n_2)}\sum_{n_1 <\ldots < n_m,  |n_1-n_m|\leq L} |\phi\{X_{a_1},n_1;X_{a_2},n_2\}$$
This yields an error of this order $O(mD^{2m}L^{m+1}/N^m)$
Similarly to our illustration the overlap of the states $\psi$ and $\zeta$ are bounded as some power of $N$. Summing over all contributions and taking into account the normalization of the m-particle state yields the claim.
\end{proof}
This Lemma mathematically implies that the Hamiltonian decouples for this basis set. In Quantum Mechanics, the stationary physical states, i.e. eigenvectors of some Hamiltonian, are orthonormal. Thus, our final step is to orthonormalize this basis. This seems quite challenging at first. We could of course use a Gram-Schmidt procedure to orthogonalize each m-particle subspace. However, this does not solve the orthogonality between each subspace. The representation used previous lemma seems too weak to finalize the result. One should for each m-particle eigenvector take a linear superposition of not only m-particle states but also 1,2,...,m-1. Solving this problem seems very hard at first.
 However, nothing forbids us to still play with another representation of the particle tensor $B^{(i)}$.
It turns out that this is the key for completing the construction which we summarize in the following theorem.
\begin{theorem}
	\label{Theorem: free Fermions}
Every injective Matrix Product State is the ground state of a hopping theory which depends solely on the transfer matrix of the Matrix Product State.
The Hamiltonian is the result of taking a long range and large system size limit of the parent Hamiltonian (\ref{eq:parentH}) with correction matrix (\ref{eq:Correction}).
A set of m-particle eigenstates in the limit, are given by the symmetric super-positions,
\begin{eqnarray}
|\psi\{X_{a_1},\ldots, X_{a_m};k_1,\ldots,k_m\}\rangle \nonumber \\
= \sum_{P\in \mathcal{S}_n}\sum_{n_1 < \ldots< n_m} e^{i(k_{P(1)}n_1+\ldots+k_{P(m)} n_m)} |\phi[B_{P(a_1)},n_{1};\ldots; B_{P(a_m)},n_{a_m}]\rangle \nonumber
\end{eqnarray}
summed over all permutations $P$.
The basis states $ |\phi[B_{a_1},n_{1};\ldots; B_{a_m},n_{a_m}]$ are given by,
\begin{eqnarray*}
&|\phi\{B^{(i)}_\alpha,j_\alpha\}_{\alpha=1}^m\rangle=\sum_{\vec{i}}\operatorname{Tr}\left(A^{i_1}\dots A^{i_{j-1}} B^{i_{n_1}}A^{i_{j+1}}\right. 
\left.\dots  B^{n_{j_m}}A^{i_{j+m}}\dots A^{i_N}\right)|\vec{i}\rangle
\end{eqnarray*}
Each tensor $B_{a_1}$ is a function of $X_a$ and $k_a$,
\begin{eqnarray*}
B^{(i)}_a=A^{(i)}X_a + e^{-i k_a}Y_aA^{(i)} - A^{(i)}Y_a,~Y_a =  (\operatorname{Id}-e^{-ik}\Gamma^*)^{-1}[X_a]
\end{eqnarray*}
The energy of each m-particle eigenvectors is the sum of the energy of the individual particle,
 \begin{eqnarray}
 \label{eq:Bethe}
 H|\psi\{X_{a_1},\ldots, X_{a_m};k_1,\ldots,k_m\}\rangle =  \sum_{j} E[X_{a_j},k_j]|\psi\{X_{a_1},\ldots, X_{a_m};k_1,\ldots,k_m\}\rangle\rangle
 \end{eqnarray}
 The particle representations $X_a$ and their respective energy are given in Lemma (\ref{Lemma: one-particle}).
 
 The states and respective energy of the theory are valid for a low-particle density which satisfies,
 $$O(D^2 |\lambda|^L)\to 0,~ O\left(\sum_{n=1}^m(m-n+1)\left(\begin{array}{c}L\\ n\end{array}\right)N^{-n}D^{2(n+1)}\right)\to 0,~ O(m!N^{-m})\to 0 $$
\end{theorem}
\begin{proof}
Let us first fix $m$, and the particles representations $X_{a_1},\ldots,X_{a_m}$ and their respective momenta $k_1,\ldots,k_m$ as given by Lemma (\ref{Lemma: one-particle}).
One should first notice the following facts.
First of all, one should see that one-particle representations with the tensor $B^{(i)}_a$ is gauge-equivalent to $X_a$,
\begin{eqnarray}
\label{eq:gauge tangent}
|\psi\{X_a;k_a\}\rangle = |\psi\{B^{(i)}_a;k_a\}\rangle
\end{eqnarray}
Secondly, the choice of gauge was chosen in such way that,
$$\forall k_1,k_2,a  \sum_{i,j} \overline{A^{(i)}}_{j,k_1}\otimes B^{(i)}_{j,k_2;a} = \sum_{i;j}\overline{B^{(i)}}_{j,k_1;a}\otimes A^{(i)}_{j,k_2} = 0 $$
It is necessary for us to compute the overlap between these new basis states $|\psi\{B^{(i)}_\alpha,k_\alpha\}_{\alpha=1}^m\rangle$,
\begin{eqnarray*}
&|\psi\{B^{(i)}_\alpha,k_\alpha\}_{\alpha=1}^m\rangle =\sum_{n_1<\ldots<n_m} e^{i\vec{k}.\vec{n}}|\phi\{B^{(i)}_\alpha,n_\alpha\}\rangle \nonumber\\ 
&|\phi\{B^{(i)}_\alpha,n_\alpha\}_{\alpha=1}^m\rangle=\sum_{\vec{i}}\operatorname{Tr}\left(A^{i_1}\dots A^{i_{j-1}} B^{i_{n_1}}_1A^{i_{j+1}}\right. 
\left.\dots  B^{n_{j_m}}_m A^{i_{j+m}}\dots A^{i_N}\right)|\vec{i}\rangle
\end{eqnarray*}
With some algebra, or by using the gauge equivalence (\ref{eq:gauge tangent}), we derive for $k = k_a = k_b$,
\begin{eqnarray*}
\sum_i \operatorname{Tr}\left(B^{(i)\dagger}_aB^{(i)}_b\right) = \operatorname{Tr}\left(X_a^\dagger\mathcal{F}_k[ X_b]\right)
\end{eqnarray*}
Using these simple properties of the tensor $B^{(i)}_a$ and particle representations $X_a$, one easily checks, 
$$\langle\psi\{B^{(i)}_,k_{a_j}\}_{j=1}^{m_1}\rangle|\psi\{B^{(i)}_{b_j},k_{b_j}\}_{j=1}^{m_2}\rangle =  \delta_{m,m_1}\delta_{m,m_2}N^{m}\prod_j^{m}\left[\delta_{k_j,k_{a_j}}\delta_{k_j,k_{b_j}}\operatorname{Tr}\left(X_{a_j}^\dagger\mathcal{F}_{k_j}[ X_{b_j}]\right)\right]+O(N^{-1})$$ 
It follows that such basis states with different particle numbers are orthogonal,

$$\frac{\langle\psi\{B^{(i)}_,k_{a_j}\}_{j=1}^{m_1}\rangle|\psi\{B^{(i)}_{b_j},k_{b_j}\}_{j=1}^{m_2}\rangle}{\sqrt{\langle\psi\{B^{(i)}_,k_{a_j}\}_{j=1}^{m_1}\rangle|\psi\{B^{(i)}_{b_j},k_{b_j}\}_{j=1}^{m_1}\rangle\langle\psi\{B^{(i)}_,k_{a_j}\}_{j=1}^{m_2}\rangle|\psi\{B^{(i)}_{b_j},k_{b_j}\}_{j=1}^{m_2}\rangle}}= O\left( \frac{N^{\min\{m_1,m_2\}}}{N^{\frac{m_1+m_2}{2}}} \right)$$
Finally, this allows us the try and find some orthonormal basis in each subspace.
 We fix the particles $X_{a_1},\ldots,X_{a_m}$ with respective momentum $k_1,\ldots,k_m$ and superimpose the basis states $\langle\psi\{B^{(i)}_{a_{P(j)}},k_{a_{P(j)}}\}_{j=1}^{m_1}|$ consisting of the permutations of the particles.
\begin{eqnarray*}
|\psi\{X_{a_1},\ldots, X_{a_m};k_1,\ldots,k_m\}\rangle \\ = \sum_{P\in \mathcal{S}_n}\sum_{n_1 <\ldots< n_m} e^{i(k_{P(1)}n_1+\ldots+k_{P(m)} n_m)} |\phi[B_{P(a_1)},n_{1};\ldots; B_{P(a_m)},n_{a_m}]\rangle
\end{eqnarray*}

We verify that this choice yields the orthogonality,
\begin{eqnarray*}
\langle\psi\{X_{a_1},\ldots, X_{a_m};k_1,\ldots,k_m\}|\psi\{X_{b_1},\ldots, X_{b_m};q_1,\ldots,q_m\}\rangle \\
\leq \frac{1}{N^m} \prod_j \delta_{a_j,b_j}\Big| \sum_{P\in S_m}\sum_{n_{P(1)}<\ldots<n_{P(m)}}e^{i\left((q_1-k_1)n_{P(1)}+\ldots+(q_m-k_m)n_{P(m)}\right)}\Big|\\
\leq \prod_j \delta_{a_j,b_j} \delta{q_j,p_j}+O\left(\frac{m!}{N^m}\right)
\end{eqnarray*}
The sought orthonormal basis has thus been found.

One, now, needs to know the action of the Hamiltonian onto the $m$-particle states. Lemma's (\ref{Lemma: one-particle}) and (\ref{Lemma: Asymptotic Regime}) can be applied. 

The theorem is proven once we manage to show that the effective Hamiltonian of each m-particle subspace decouples,
$$H^{(m)}\{k_1,\ldots,k_m\} = \sum_j H^{(1)}\{k_j\}$$ 
Our choice of orthogonal basis indeed implies such result,
\begin{eqnarray*}
\langle\psi\{X_{a_1},\ldots, X_{a_m};k_1,\ldots,k_m\}|&H|\psi\{X_{b_1},\ldots, X_{b_m};q_1,\ldots,q_m\}\rangle\\
 &= \sum_{j_1,j_2} \prod_{j\not\in\{j_1,j_2\}} \delta_{a_j,b_j} \delta{q_j,p_j} \langle\psi\{X_{a_{j_1}};k_{j_1}\}|H|\psi\{X_{b_{j_2}};q_{j_2}\}\rangle 
\end{eqnarray*} 
With this the theorem has been shown.

\end{proof}
We showed the symmetrized states could be considered as a type of eigenstates of the Hamiltonian. Unfortunately, this is only a small part of the spectrum. From our results in the previous lemma's, we know that the $n$-particle spectrum is some complex superposition of the $n$-particle basis. However, in the case of normal unital quantum channels, this superposition turns out to be extremely simple. All $n$-particle energy levels are degenerate and a trivial basis can be found by simply taking the symmetric and anti-symmetric superposition of the $n$-particle basis.
\begin{corollary}
	\label{Corollary: Exactly Solvable}
	When the transfer matrix is a normal unital quantum channel, in the long range and system size limit, the Hamiltonian (\ref{eq:parentH}) becomes exactly solvable with eigenstates

	\begin{eqnarray}
	|\psi\{X_{a_1},\ldots, X_{a_m};k_1,\ldots,k_m\}\rangle \nonumber \\
	= \sum_{P\in \mathcal{S}_n}\sum_{n_1 < \ldots< n_m} e^{i(k_{P(1)}n_1+\ldots+k_{P(m)} n_m)} A(P)(X_{a_1},\ldots,X_{a_m})|\phi[B_{P(a_1)},n_{1};\ldots; B_{P(a_m)},n_{a_m}]\rangle \nonumber
	\end{eqnarray}
	with $A(T_{\alpha,\beta}P)=\pm A(P) $.
\end{corollary}

\subsection{Normal unital trace-preserving Completely Positive Operators}
In general, we need not only information about the spectral properties of the transfer matrix, but also the algebraic relation between the eigenvectors.

It turns out that normal unital trace preserving completely positive operators are an interesting class to study, as a lot of the computations, we had to endure for proving Theorem (\ref{Theorem: free Fermions}), simplifies.

$$\rho=\frac{\mathbb{1}}{\operatorname{Tr}(\mathbb{1})},~~\Gamma^*[\mathbb{1}]=\mathbb{1},~~\Gamma[\mathbb{1}]=\mathbb{1} $$

It is important to stress that not any spectrum is allowed. Some, sometimes heavy, restrictions are imposed due to the complete positiveness.
For example, if we assume that the eigenvectors have a particular group structure
\begin{eqnarray}
\label{eq:unitspec}
X_a=U_g,~~U_g=\oplus_j e^{\frac{2\pi ijg}{D}}
\end{eqnarray}
This example is based on the lemma that spectral vectors around unit circle are of this form (\ref{eq:unitspec}).

Still, the following restrictions can be shown for $\Gamma[U_g]=|\lambda_g| e^{\frac{2\pi i \kappa_g}{D}} U_g$,
\begin{eqnarray}
\label{eq:normalSpec}
\forall \alpha\in\{1,\ldots,2D\},~~ \sum_g |\lambda_g| \cos\left(\frac{2\pi \left[\kappa_g-g\alpha\right]}{D}\right) \geq 0
\end{eqnarray}

Writing, 
$$\Gamma[.]=\sum_g \lambda_g U_g \operatorname{Tr}\left(U_g^\dagger[.]\right)$$
Equation (\ref{eq:normalSpec}) follows for the necessary and sufficient for complete positiveness \cite{chapter 1: Choi CP},
$$C[\Gamma]=\sum_{\alpha,\beta}|\alpha\rangle\langle \beta|\otimes \Gamma[|\alpha\rangle\langle \beta|]\geq 0 $$
This simple example illustrates that spectral properties of completely positives are not just arbitrary. There are physics restrictions. It is amazing that nonetheless, we can connect these with the spectrum of Hamiltonians, which are pretty much arbitrary.

As presented at the end of the work, the appearance of bound states, is speculated from the breaking down of the convergence towards the one-particle subspace.
\begin{eqnarray}
\label{eq:error1Part}
\lim_{L\to\infty}\epsilon_1(X_a,k)=\lim_{L\to\infty}\frac{\|H|\psi\{X_a,k\}\rangle-2(L-\sigma_{k,a})|\psi\{X_a,k\}\rangle\|_2^2 }{\langle\psi\{\overline{X_a},k\}|\psi\{X_a,k\}\rangle}>0\\
\sigma_{k,a}=\operatorname{Re}\frac{1}{1-e^{i\left[k-\phi_{X_a}\right]}|\lambda_{X_a}|}
\end{eqnarray}

The following conclusions can be taken.
If for all momenta $k$, $\epsilon_1(X_a,k)\to 0$, then  $|\psi\{ X_a,k\}\rangle$ are eigenstates with energy $\operatorname{Re}1/(1-e^{i[k-\phi_{X_a}]}|\lambda_{X_a}|)$.

In the case of bound states lying in the $2-$ and $1-$ particle, i.e. $|\psi\{X_a,k\}\rangle +\sum_q c_q |\psi\{X_\alpha,k-q;X_\beta^\dagger,q\}$, we need a convergence of the 3-particle errors,
$$\epsilon(X_\alpha,q),\epsilon(X_\beta,q)\to 0 $$

\newpage

\end{document}